\documentclass[aps,prl,twocolumn,unsortedaddress]{revtex4-2}

\usepackage{graphicx}
\usepackage{dcolumn}
\usepackage{bm}
\usepackage{lipsum}
\usepackage[T1]{fontenc}
\usepackage{mathpazo}

\usepackage{dsfont}
\usepackage{enumitem}
\usepackage{tcolorbox}
\usepackage{needspace}

\usepackage{amsthm} 
\usepackage{thmtools} 
\usepackage{thm-restate}
\usepackage{amsmath}
\usepackage{amssymb}
\usepackage{amsfonts}
\usepackage{mleftright}
\usepackage{mathtools}
\usepackage{derivative}
\usepackage{diagbox}

\usepackage[normalem]{ulem}
\usepackage[colorlinks,linkcolor=blue,urlcolor=blue, citecolor=blue]{hyperref}
\usepackage{cleveref}
\usepackage[ruled,procnumbered]{algorithm2e}
\usepackage{algpseudocode}

\declaretheorem[name=Theorem,
  refname={theorem,theorems},
  Refname={Theorem,Theorems}]{theorem}
\declaretheorem[numberlike=theorem,name=Lemma,
  refname={lemma,lemmas},
  Refname={Lemma,Lemmas}]{lemma}
\declaretheorem[numberlike=theorem,name=Proposition,
  refname={proposition,propositions},
  Refname={Proposition,Propositions}]{proposition}
\declaretheorem[numberlike=theorem,name=Corollary,
  refname={corollary,corollaries},
  Refname={Corollary,Corollaries}]{corollary}
\declaretheorem[numberlike=theorem,name=Definition,
  refname={definition,definitions},
  Refname={Definition,Definitions}]{definition}
\declaretheorem[numberlike=theorem,name=Assumption,
  refname={assumption,assumptions},
  Refname={Assumption,Assumptions}]{assumption}
\declaretheorem[numberlike=theorem,name=Remark,
  refname={remark,remarks},
  Refname={Remark,Remarks}]{remark}
\crefname{algocf}{protocol}{protocols}
\Crefname{algocf}{Protocol}{Protocols}

\usepackage{xifthen}

\newcommand{\trace}[0]{\mathrm{tr}}
\newcommand{\ket}[1]{\mathinner{|{#1}\rangle}}

\newcommand{\ketbra}[2]{\mathinner{|{#1}\rangle\langle{#2}|}}
\newcommand{\vecket}[1]{\mathinner{|{#1}}\rangle\!\rangle}
\newcommand{\vecbra}[1]{\mathinner{\langle\!\langle{#1}|}}
\newcommand{\vecketbra}[2]{\mathinner{|{#1}\rangle\!\rangle\langle\!\langle{#2}|}}
\newcommand{\vecbraket}[2]{\mathinner{\langle\!\langle{#1}|{#2}\rangle\!\rangle}}

\newcommand{\mathdynact}[0]{\mathfrak{A}}
\newcommand{\mathrestim}[0]{\mathfrak{T}}
\newcommand{\mathmerit}[0]{\mathfrak{S}}
\newcommand{\accuracyN}[0]{\mathfrak{N}}

\newcommand{\Lindblad}[3]{
    \ifthenelse{\isempty{#1}}%
        {\ifthenelse{\isempty{#2}}
            {\ifthenelse{\isempty{#3}}%
                {\mathcal{L}\mleft[\bullet\mright]}%
                {\mathcal{L}(#3)\mleft[\bullet\mright]}%
            }
            {\ifthenelse{\isempty{#3}}%
                {\mathcal{L}\mleft[#2\mright]}%
                {\mathcal{L}(#3)\mleft[#2\mright]}%
            }
        }
        {\ifthenelse{\isempty{#2}}
            {\ifthenelse{\isempty{#3}}%
                {\mathcal{L}_{#1}\mleft[\bullet\mright]}%
                {\mathcal{L}_{#1}(#3)\mleft[\bullet\mright]}%
            }
            {\ifthenelse{\isempty{#3}}%
                {\mathcal{L}_{#1}\mleft[#2\mright]}%
                {\mathcal{L}_{#1}(#3)\mleft[#2\mright]}%
            }
        }
    }

\newcommand{\dissipator}[2]{
    \ifthenelse{\isempty{#2}}%
        {\mathcal{D}\mleft[#1\mright][\bullet]}%
        {\mathcal{D}\mleft[#1\mright]\mleft[#2\mright]}%
    }

\newcommand{\jumpalphabet}[1]{
    \ifthenelse{\isempty{#1}}%
        {\mathcal{J}}%
        {\mathcal{J}_{#1}}%
    }

\newcommand{\memorytorates}[1]{
    \ifthenelse{\isempty{#1}}%
        {\gamma}
        {\gamma^{(#1)}}
    }

\newcommand{\memoryupdate}[3]{
    \ifthenelse{\isempty{#1}}%
        {\ifthenelse{\isempty{#2}}
            {\upsilon_{#3}}
            {\upsilon_{#3}(#1,#2)}
        }
        {\upsilon_{#3}(#1,#2)}
    }

\newcommand{\parameterspace}[1]{
    \ifthenelse{\isempty{#1}}%
        {\mathcal{G}}%
        {\mathcal{G}^{#1}}%
    }

\newcommand{\memoryspace}[0]{\mathcal{M}}

\newcommand{\hilmap}{\mathcal{H}}

\newcommand{\jidx}[0]{j}
\newcommand{\altjidx}[0]{k}
\newcommand{\sidx}[0]{a}
\newcommand{\sidxprime}[0]{b}
\newcommand{\cidx}[0]{x}
\newcommand{\altcidx}[0]{y}
\newcommand{\calphabet}[1]{\mathcal{X}_{#1}}
\newcommand{\ridx}[0]{k}
\newcommand{\Ridx}[0]{K}
\newcommand{\numpps}[0]{G}
\newcommand{\matLindblad}[1]{\widehat{\mathcal{L}}_{#1}}

\newcommand{\weights}[1]{w_{#1}}
\newcommand{\idxset}[0]{\mathcal{I}}

\newcommand{\textCWI}[0]{clockwork independence}

\newcommand{\textST}[0]{self-timing}

\newcommand{\expect}[2]{
    \ifthenelse{\isempty{#1}}
        {\ifthenelse{\isempty{#2}}
            {\mathbb{E}}
            {\mathbb{E}\mleft[#2\mright]}
        }
        {\ifthenelse{\isempty{#2}}
            {\mathbb{E}_{#1}}
            {\mathbb{E}_{#1}\mleft[#2\mright]}
        }
}

\newcommand{\var}[2]{
    \ifthenelse{\isempty{#1}}
        {\ifthenelse{\isempty{#2}}
            {\mathrm{Var}}
            {\mathrm{Var}\mleft[#2\mright]}
        }
        {\ifthenelse{\isempty{#2}}
            {\mathrm{Var}_{#1}}
            {\mathrm{Var}_{#1}\mleft[#2\mright]}
        }
}

\usepackage[page]{appendix}

\renewcommand{\appendixtocname}{List of appendices}

\makeatletter
\let\oldappendix\appendices

\g@addto@macro\tableofcontents{%
  \let\tf@toc@orig\tf@toc
}
\renewcommand{\appendices}{%
  \renewcommand{\thesection}{}
  \let\tf@toc\tf@app
  \addtocontents{app}{\protect\setcounter{tocdepth}{1}}
  \immediate\write\@auxout{%
    \string\let\string\tf@toc\string\tf@app
  }
  \oldappendix
}%

\g@addto@macro\endappendices{%
  \let\tf@toc\tf@toc@orig
  \immediate\write\@auxout{%
    \string\let\string\tf@toc\string\tf@toc@orig
  }%
}  

\renewcommand\tableofcontents{%
    \@starttoc{toc}%
}

\newcommand{\listofappendices}{%
  \begingroup
  \newcommand{\contentsname}{\appendixtocname}
  \let\@oldstarttoc\@starttoc
  \def\@starttoc##1{\@oldstarttoc{app}}
  \tableofcontents
  \endgroup
}

\makeatother

\begin{document}

\preprint{APS/123-QED}

\title{Feedback-Induced Advantage in Quantum Clockworks}

\author{Jakob Miller}
\email[]{jakob.miller@math.ku.dk}
\affiliation{Department of Mathematical Sciences, University of Copenhagen, Universitetsparken 5, 2100 Denmark}
\affiliation{Institute for Theoretical Physics, ETH Z\"urich, Wolfgang-Pauli-Strasse 27, 8093 Z\"urich, Switzerland}
\affiliation{Atominstitut, Technische Universität Wien, 1020 Vienna, Austria}


\author{Paul Erker}
\email[]{paul.erker@tuwien.ac.at}
\affiliation{Atominstitut, Technische Universität Wien, 1020 Vienna, Austria}
\affiliation{IQOQI Vienna, Austrian Academy of Sciences, 1090 Vienna, Austria}

\date{\today}

\begin{abstract}
Atomic frequency standards have achieved steadily increasing precision over the past seventy years, enabled in part by feedback mechanisms that stabilise their output. In parallel,  the timekeeping capabilities of quantum systems have been explored within the recently developed ticking-clock framework, which models clocks as dynamical systems producing a stochastic sequence of ticks. However, a theoretical description that unifies these perspectives and incorporates feedback into autonomous quantum clocks has been lacking.
We introduce a framework for feedback-controlled clockworks in which classical information extracted from the tick sequence is used to influence the subsequent dynamics of the clock. We show that such feedback preserves the core structural features of \textST{} and \textCWI{} that characterise autonomous ticking clocks. We further identify the signal-to-noise ratio $\mathmerit$ as the fundamental figure of merit for assessing the performance of feedback-controlled clocks.
Applying our framework to two representative architectures, we prove that classical clockworks cannot surpass the optimal signal-to-noise ratio achievable without feedback. In contrast, for quantum clockworks we present numerical evidence that feedback can provide a genuine and robust performance enhancement, improving the maximal attainable signal-to-noise ratio. We provide additional numerical evidence, suggesting that the magnitude of the improvement strongly depends on the amount and type of the stored, classical information. These results establish feedback as a potentially essential ingredient in pushing the fundamental limits of timekeeping in the quantum regime.
\end{abstract}

\maketitle

\paragraph*{Introduction.}

 The development of quantum technologies has renewed interest in understanding how timekeeping should be modeled at a fundamental level. Time plays a special role in quantum theory, entering as an external parameter rather than as an observable~\cite{Pauli_time,time_in_physics}. Thus, any time reading must be produced by the dynamics of a physical clock. In classical settings, constructing reliable clocks is largely an engineering challenge and a fully self-contained theoretical description of the clock as well as its control mechanisms is usually unnecessary. In contrast, any mechanism in a quantum system producing a classical time reading will non-trivially influence the clock's dynamics and therefore must be modeled explicitly in any self-contained theoretical description~\cite{Malabarba_2015,Woods2018,Milburn2020,Woods2021,Xuereb2023,Guzman2024,aqpu,Wadhia2025,Campbell2026}. Since control mechanisms whose dynamics are not explicitly modeled can act as hidden timing resources, a complete model of quantum timekeeping should additionally include the energetic and dynamical effects of the feedback mechanisms that stabilize and regulate the clock. Feedback control of quantum systems is well established~\cite{quantum_feedback_review,ped_intro_mmt_based_feedback} and deterministic equations for jump-based feedback in open quantum systems have recently been derived~\cite{rosal2025deterministicequationsfeedbackcontrol1,rosal2025deterministicequationsfeedbackcontrol2,rosal2025deterministicequationsfeedbackcontrol3}. Quantum engineering has also been shown to improve the phase stability of an opto-electronic feedback oscillator~\cite{loughlin2026quantumenhancedfeedbackoscillator}, although there feedback refers to an amplifier loop rather than to control conditioned on observed ticks. For ticking clocks, specific externally controlled protocols have been proposed~\cite{Barato_Seifert,yang2019accuracyenhancingprotocolsquantum}. Despite their central importance, however, feedback processes have not yet been incorporated into a general theoretical description of autonomous quantum clocks, in which the control is part of the clock.
 

Recent work has clarified the fundamental limits on clocks viewed as machines whose task it is to generate a regular sequence of \textit{ticks}. Both information-theoretic and thermodynamic constraints have been established. Clock performance is commonly characterized by the accuracy $\accuracyN$, the number of reliable ticks, and the resolution $\nu$, the average tick rate. For clocks undergoing incoherent dynamics, the second law bounds the accuracy to grow at most linearly with the entropy produced per tick~\cite{Erker2017, Milburn2020}. In contrast, coherent quantum dynamics can yield an accuracy that scales exponentially with entropy production~\cite{Meier2024a}. Moreover, considering the accessible Hilbert-space dimension as a resource reveals that quantum clocks achieve a quadratic accuracy advantage over classical clocks living in a discrete state space of equal dimension~\cite{Woods2022}. An accuracy--resolution tradeoff has also been identified: increasing accuracy reduces resolution~\cite{Meier2023}, with quantum clocks again offering a quadratic improvement $\accuracyN\sim\nu^{-2}$ over the classical limit $\accuracyN\sim\nu^{-1}$. 
Furthermore, for clockworks undergoing incoherent dynamics, the optimal choice of processing the time information obtained from the clockwork has been established by introducing a new thermodynamic uncertainty relation that tightly bounds the signal-to-noise ratio $\mathmerit=\accuracyN\nu$ \cite{Prech2025}.
More and more of these theoretical predictions are being tested experimentally. The link between entropy production and accuracy was first demonstrated in a classical electromechanical resonator~\cite{Pearson2021}. Quantum measurement backaction on a clock has been probed in superconducting circuits~\cite{He2023}. More recently, entropy dissipation associated with tick readout was measured in a semiconductor quantum-dot clock, identifying this as the dominant thermodynamic cost at the quantum scale~\cite{Wadhia2025}.

 A major open challenge is understanding timekeeping at the precision frontier.  Atomic clocks provide the most accurate time measurements to date~\cite{Aeppli2024} and play an essential role in tests of fundamental physics. Their operation relies on feedback: a control loop steers a local oscillator towards an atomic transition using error signals accumulated over many past interrogation cycles~\cite{allan_manual,smart_optical_lattice_2014}. However, no theoretical framework for autonomous clocks currently describes such feedback-stabilization, frequency downconversion and other control processes that are integral to atomic clocks. Given that state-of-the-art atomic clocks surpass classical thermodynamic precision bounds~\cite{Pietzonka2024}, they may already exploit quantum coherence for enhanced energy efficiency~\cite{Meier2024a}, but how exactly these mechanisms can be modeled from first principle remains to be understood. Clocks possess no special thermodynamic status. Precisely for this reason, the limits on their performance must originate in concrete, identifiable mechanisms, such as their control loops, which should therefore be modeled explicitly. This raises the question of whether feedback can improve the precision of an autonomous clock at all and, if so, how much stored information the controller needs and whether the advantage survives decoherence.


Here we tackle this question by posing a self-contained framework for feedback in ticking clocks, in which classical control with finite memory adjusts the clockwork dynamics based on the observed ticks, without acting as a hidden timing resource. We quantify performance by the signal-to-noise ratio (SNR), which, unlike other measure as for example the accuracy \cite{Erker2017,Prech2025}, is unaffected by rescaling the tick count.  We show that feedback cannot improve classical clockworks, whereas for two qubit clockworks even in the simplest case it raises the maximal SNR by about 9\%, an advantage that is robust to weak dephasing and can grow with the controller's memory. We begin by reviewing the framework of ticking clocks.

\begin{figure}[]
    \centering
    \includegraphics[width=1\linewidth]{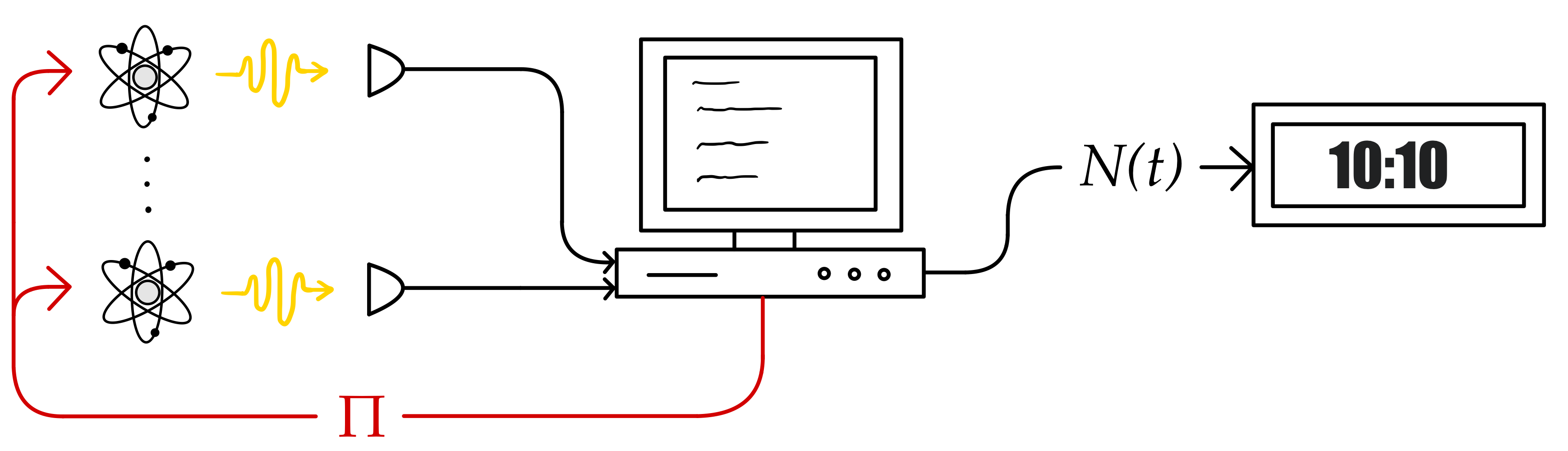}
    \caption{Schematic representation of our feedback framework. Jumps in the (quantum) clockworks are observed by detectors, connected to the classical control unit. Based on the observed pattern of jumps, the control unit applies feedback to the clockworks, according to the feedback policy $\Pi$ (see \cref{def:feedback_policy}), and calculates an estimate of the time $N(t)$, that is displayed in the tick register.}
    \label{fig:feedback_scheme}
\end{figure}

\paragraph*{Ticking Clocks.}
The framework of ticking (quantum) clocks provides an axiomatic approach to classify the physical systems that we consider to be clocks \cite{Woods2021, RNH_ticking_clocks}. Below, we briefly summarize the relevant concepts and terminology following \cite{RNH_ticking_clocks}.Motivated by our everyday interactions with clocks, we require that neither the clock system nor the future time reading of the clock are disturbed by the act of reading off the time.
This \emph{independence principle} clearly separates ticking clocks from general quantum clocks whose state would be altered upon accessing the time information, e.g., by measuring the system 
\footnote{So-called stopwatches would be an example of this type of quantum clock \cite{Woods2022}.}.
One can show that the independence principle implies that every ticking clock is equivalent to a system, called the minimal ticking clock, that can be decomposed into two separate parts: 
The first part, which we will call the clockwork from hereon, is responsible for the time evolution of the system.
The second part, which we will call the (tick) register, stores the clock's estimate of the background time and is accessed when reading off the clock time
\footnote{The equivalence holds up to transformations under a CPTP map, that is reversible and fixed for all times. Due to this equivalence we restrict our attention to minimal ticking clocks and use the terms ticking clock and minimal ticking clock synonymously.}.
In particular, a minimal ticking clock has an associated Hilbert space $\hilmap$, density operator $\rho(t)$ and projectors $\Pi_m$ that decompose as%
\begin{align}
\begin{split}
    \hilmap &= \bigoplus_{n\in\idxset} \hilmap_{C_n},\\
    \rho(t) &= \bigoplus_{n\in\idxset} \rho_{C_n}^{(n)}(t), \label{eq:minimal_ticking_clock} \ \\
    \Pi_m &= \bigoplus_{n\in\idxset} \delta_{m,n} \mathbb{1}_{C_n},
\end{split}
\end{align}
where $\rho^{(n)}_{C_n}(t)$ denote non-normalized density operators and $\mathbb{1}$ denotes the identity.
The operators $\Pi_m$ are the projectors onto the $m$\textsuperscript{th} subspace of the Hilbert space $\hilmap$.
In this decomposition, the clockwork and register can be identified as follows:
The indices $n,m\in\idxset$ take the role of the register.
They encode or register all information about the clock's estimate of the current background time and are, hence, called register indices.
The probability of reading off the value $m$
of the clock's display at time $t$ is given by $\trace[\Pi_m\rho(t)]=\trace[\rho^{(m)}_{C_m}(t)]$.
Therefore, performing a measurement associated to the POVM $\{\Pi_m\}_{m\in\idxset}$ can be interpreted as accessing the clock's current time information.
The Hilbert space $\hilmap_{C_n}$, on the other hand, is the space in which the clockwork lives. 
Note, that in general this can be a different (physical) clockwork for each $n$.\\

Silva et al. \cite{RNH_ticking_clocks} propose four different properties that impose further restrictions on the set of ticking clocks. 
For us, the relevant properties are \emph{\textST{}} and \emph{\textCWI{}}.
We say that the clock satisfies \textCWI{}, if the mathematical structure and dynamics of the clockwork are independent of the register index $n\in\idxset$, that encodes the clock's estimate of the time.
Intuitively, this captures that the clock uses the same clockwork to produce all of its ticks.
Under this assumption, \cref{eq:minimal_ticking_clock} simplifies to
\begin{align}
\begin{split}
    \hilmap &= \hilmap_C \otimes\hilmap_{T}\\
    \rho(t) &= \sum_{n\in\idxset} \rho_{C}^{(n)}(t) \otimes \ketbra{n}{n}_T\\
    \Pi_m &= \mathbb{1}_{C}\otimes\ketbra{m}{m}_T,
\end{split}
\end{align}
where the Hilbert space $\hilmap_T$ is called the tick register.
We say that a ticking clock fulfills \textST{} if the time evolution of the clock is independent of its environment and Markovian, i.e., if the evolution of the clock only depends on the initial state of the clock $\rho(0)$ at time $t=0$ and if the evolution is homogeneous in time.
This ensures that the clock does not receive any time information from its environment and therefore does not rely on some form of external timing resources.
For a ticking clock satisfying both \textST{} and \textCWI{}, the time evolution is governed by a quantum Markovian master equation
\begin{equation}
    \odv{}{t}\rho(t) = \Lindblad{}{\rho(t)}{}
    ,
\end{equation}
with a Lindblad operator
\begin{equation}
    \Lindblad{}{}{} = -i[H,\bullet] + \sum_{\jidx}\dissipator{J_{\jidx}}{}.
\end{equation}
We defined dissipators $\dissipator{J}{} = J\bullet J^\dagger - \frac{1}{2} \{J^\dagger J,\bullet\}$, the Hamiltonian $H = H_C\otimes\mathbb{1}_T$ and the jump operators $J_{\jidx,n} = J_{C,\jidx}\otimes\ketbra{n+\weights{\jidx}}{n}_T$,
where the index $\jidx\in\jumpalphabet{}$ denotes the type of jump that occurred and $\jumpalphabet{}$ denotes the set of all possible types of jumps.
The weight $\weights{\jidx}$ determines how the register index $n$ changes, given that a jump of type ${\jidx}$ occurred.

Observe, that the operators $H_C$ and $J_{C,\jidx}$ are responsible for the Hermitian and non-Hermitian time evolution of the clockwork, respectively, whereas the tick register $T$ does not participate in the time evolution \cite{RNH_ticking_clocks}. Its sole function is to store a weighted count $N(t)=\sum_{\jidx}\weights{\jidx}N_{\jidx}(t)$ of the number of jumps that have occurred until time $t$, where $N_{\jidx}(t)$ denotes the number of jumps of type $\jidx$ that have occurred until time $t$. 
In particular, for each jump of type $\jidx$ the count $N(t)$ increases by an increment $\weights{\jidx}\in\mathbb{R}$.
More formally, the tick register stores the value that an \emph{integrated current} with weights $\weights{\jidx}$ takes on.
For a more detailed description of integrated currents and their properties in the asymptotic infinite time limit, that can be obtained using the tools of full counting statistics, we refer to \cref{app:FCS_for_clocks} or \cite{bridging_the_gap}.
Due to this fact, one may neglect the tick register and describe a ticking clock satisfying \textST{} and \textCWI{} by a reduced Lindblad operator, specified by the Hamiltonian $H_C$ and the jump operators $J_{C,\jidx}$, together with an integrated current, that is specified by the weights $\weights{\jidx}$.
In this formalism, linear post-processing of the time signal is included into the choice of weights $\{\weights{\jidx}\}_{\jidx}$. Choosing a new set of weights (i.e. a new integrated current) corresponds to a different way of post-processing the information in $\{N_{\jidx}(t)\}_{\jidx}$. \\

\paragraph{Clock precision.}
We use the signal-to-noise ratio (SNR), defined as 
\begin{equation}
    \label{eq:SNR}
        \mathmerit(t) \coloneq \frac{\mleft(\odv{}{t}\expect{t}{N}\mright)^2}{\var{t}{N}/t},
\end{equation}
to measure the quality of the time signal produced by a ticking clock.
Here, $\expect{t}{N}$ and $\var{t}{N}$ denote the expectation value and variance of the random variable $N(t)$ at background time $t$.
The SNR clearly depends on the choice of weights $\{\weights{\jidx}\}_\jidx$ specifying the integrated current $N(t)$ and as a measure of the precision of a quantum clock has recently been studied in \cite{Prech2025}.
For a discussion of its relation to the accuracy $\accuracyN$ we refer to \cref{app:clock_quality}.
Note, that throughout this work we restrict our attention to the asymptotic longtime limit of the SNR $\mathmerit\coloneq\lim_{t\rightarrow\infty}\mathmerit(t)$. 
Hence, we will always assume the clockwork to be in a steady state, i.e., a state that lies in the kernel of the reduced Lindblad operator.\\

\paragraph*{Feedback Framework.}
Before we introduce our model of feedback, consider two ticking clocks satisfying \textST{} and \textCWI{} that consist of clockworks $C_{\sidx}$ and tick registers $T_{\sidx}$ for $\sidx=1,2$. When having access to both clocks, one could read off both time estimates $N^{(\sidx)}(t)$ stored in the tick registers $T_{\sidx}$ and apply linear post-processing $N(t)=\sum_{\sidx=1,2} \alpha_\sidx N^{(\sidx)}(t)$, in order to obtain an estimate of time that is better than that given by each $N^{(\sidx)}(t)$, individually.
Observe, that this $N(t)$ corresponds to an integrated current of the joint system $\hilmap_{C}=\hilmap_{C_1}\otimes\hilmap_{C_2}$, whose time evolution is governed by the joint Lindblad operator $\mathcal{L}_{C_1}\otimes\mathcal{I}_{C_2}+\mathcal{I}_{C_1}\otimes\mathcal{L}_{C_2}$, where $\mathcal{I}$ is the identity channel. Hence, we can view two independent ticking clocks together with post-processing as a single, large ticking clock, with clockwork space $\hilmap_C$ and tick register $T$, storing the value that $N(t)$ takes on
\footnote{Of course, one can always choose to ignore the information obtained from clockwork $C_1$ by setting $\alpha_{0}=1$ and $\alpha_{1}=0$ to recover $N(t)=N^{(0)}(t)$ and vice versa.}. 

Inspired by Markov decision processes \cite{puterman2014markov}, the core idea behind our feedback model is to expand the bipartite structure, comprised of clockwork and tick register (possibly representing multiple ticking clocks plus post-processing), to a tripartite structure which additionally includes a classical control unit, living in Hilbert space $\hilmap_{M}.$ As the name suggests the control unit exercises control over the clockworks, with its decisions based solely on the pattern of jumps that have occurred in the joint clockwork $C$.
Since the control unit remains classical at every point in time we refer to our model as \emph{incoherent feedback}. A schematic illustration of this model can be found in \cref{fig:feedback_scheme}. 

In order for the overall system consisting of the joint clockwork $C$, tick register $T$ and control unit $M$ to satisfy the intuition behind \textST{} and \textCWI{}, we make the following four assumptions:
Firstly, we assume that the control unit exercises its control  over clockwork $C_{\sidx}$ by choosing between different sets of parameters $\mathbf{c}\in\parameterspace{\sidx}$ specifying the Lindblad operator 
\begin{equation}
        \Lindblad{C_{\sidx}}{}{\mathbf{c}} = -i[H_{C_{\sidx}}(\mathbf{c}),\bullet] + \sum_{\jidx}\dissipator{J_{C_{\sidx},\jidx}(\mathbf{c})}{}
    \label{eq:Lindblad_with_control}    
\end{equation}
and hence future time evolution of clockwork $C_\sidx$. 
Here $\parameterspace{\sidx}$ denotes the set of possible combinations of control parameters. The kind of control over the dynamics of system $C_{\sidx}$ is determined by the dependence of the Hamiltonian $H_{C_{\sidx}}(\mathbf{c})$ and jump operators $J_{C_{\sidx},\jidx}(\mathbf{c})$ on $\mathbf{c}$.
Secondly, we assume that a change of the parameters $\mathbf{c}$ is implemented instantaneously and without altering the clockworks' state in the process.
Thirdly, in order to ensure that the control unit does not rely on any additional timing resources, we assume that the parameters $\mathbf{c}$ can only be set immediately after a jump in one of the clockworks $C_\sidx$ has occurred and that the choice of parameters is fully determined by the state of the control unit's memory $m$, which encodes information about previous decisions and the sequence of jumps that have so far occurred in the joint clockwork $C$. 
Lastly, we a assume the memory to be of finite size.
A motivation and detailed description of these assumptions, which lead us to the definition of what we call a \emph{feedback policy}, can be found in \cref{app:assumptions_feedback_policy}.

\begin{definition}[Feedback policy]\label{def:feedback_policy}
    For systems $C_{\sidx}$ with a Lindblad operator as in~\cref{eq:Lindblad_with_control} for $\sidx=1,\ldots,{\numpps}$, a feedback policy $\Pi$ is a tuple $(\memoryspace, \memoryupdate{}{}{}, \memorytorates{})$ consisting of
    \begin{enumerate}[label=(\roman*)]
        
        \item a finite set $\memoryspace$, called the memory state space,
        
        \item a function $\memoryupdate{}{}{}: \mathcal{M}\times\mleft(\bigcup_{\sidx=1}^{\numpps}\jumpalphabet{\sidx}\mright)\rightarrow  \mathcal{M}$, called the memory-update function, and
        
        \item a set of functions $\memorytorates{} = \{\memorytorates{\sidx}\}_{\sidx=1}^{\numpps}$, where $\memorytorates{\sidx}: \memoryspace \rightarrow \parameterspace{\sidx}$, is called the parameter-update function for $C_{\sidx}$.
    \end{enumerate}
\end{definition}

Since the control unit's memory remains classical at all times, we may define a Hilbert space $\hilmap_M$ of dimension $\mathrm{dim}(\hilmap_M)=|\memoryspace|$ and an orthonormal basis $\{\ket{m}_M\}_{m\in\memoryspace}$ such that the control unit is described by a diagonal density operator in $\hilmap_M$. Then, an update of the memory $m\rightarrow\memoryupdate{m}{\jidx}{}$ upon observing a jump of type $\jidx$ can be viewed as a jump operator $\ketbra{\memoryupdate{m}{\jidx}{}}{m}_M$. The Lindblad operator $\Lindblad{\mathrm{fb}}{}{}$ governing the dynamics of system $\hilmap_{C_1}\otimes\cdots\otimes\hilmap_{C_{\numpps}}\otimes\hilmap_M$
is given by
\begin{align}
\begin{split}
    \label{eq:Lindblad_evo_conditioned_on_memory}
    \Lindblad{\mathrm{fb}}{}{} = \sum_{m\in\memoryspace} \sum_{\sidx=1}^{\numpps}\biggl( &-i[H_{C_a M}(m)\otimes\mathbb{1}_{C_{\bar{\sidx}}},\bullet] + \\
    &+ \sum_{\jidx\in\jumpalphabet{\sidx}} \dissipator{J_{C_a M,\jidx}(m)\otimes\mathbb{1}_{C_{\bar{\sidx}}}}{} \biggr),
\end{split}
\end{align}
where $\mathbb{1}_{C_{\bar{\sidx}}}=\mathbb{1}_{C_1}\otimes\cdots\otimes\mathbb{1}_{C_{\sidx-1}}\otimes\mathbb{1}_{C_{\sidx+1}}\otimes\cdots\otimes\mathbb{1}_{C_{\numpps}}$ and 
\begin{equation}
    \label{eq:Hamiltonian_conditioned_on_memory}
    H_{C_a M}(m) = H_{C_a}(\memorytorates{\sidx}(m))\otimes\ketbra{m}{m}_M
\end{equation} 
ensures that the unitary time evolution of system $C_{\sidx}$ is determined by $H_{C_a}(\memorytorates{\sidx}(m))$, given that the control unit's memory is in state $\ket{m}_M$, and where
\begin{equation}
    \label{eq:Jump_op_conditioned_on_memory}
    J_{C_a M,\jidx}(m) = J_{C_a,\jidx}(\memorytorates{\sidx}(m))\otimes\ketbra{\memoryupdate{m}{\jidx}{}}{m}_M
\end{equation}
describes a jump of type $\jidx$ (in system $C_{\sidx}$), given that the control unit's memory was in state $\ket{m}_M$ before the jump, followed by an immediate update of the control unit's memory to $\ket{\memoryupdate{m}{\jidx}{}}_M$.
Note, that the terms $(J_{C_a M,\jidx}(m))^\dagger J_{C_a M,\jidx}(m) = (J_{C_a,\jidx}(\memorytorates{\sidx}(m)))^\dagger J_{C_a,\jidx}(\memorytorates{\sidx}(m)) \otimes \ketbra{m}{m}_M$ that appear in $\dissipator{J_{C_a M,\jidx}(m)\otimes\mathbb{1}_{C_{\bar{\sidx}}}}{}$ have the same form as~\cref{eq:Hamiltonian_conditioned_on_memory}.
Thus, the evolution under $\Lindblad{\mathrm{fb}}{}{}$ indeed preserves the control unit's classical structure.
Treating $\hilmap_{C_1}\otimes\cdots\hilmap_{C_G}\otimes\hilmap_M$ as one large clockwork, it satisfies \textST{} and \textCWI{}, whereas a single subsystem $\hilmap_{C_{\sidx}}$ in general fulfills neither of the two properties.
In \cref{app:previous_work}, we show how our framework generalizes previous models of feedback and external control in (quantum) clocks \cite{yang2019accuracyenhancingprotocolsquantum, Barato_Seifert}.

\begin{figure}
    \centering
    \includegraphics[width=1.15\linewidth]{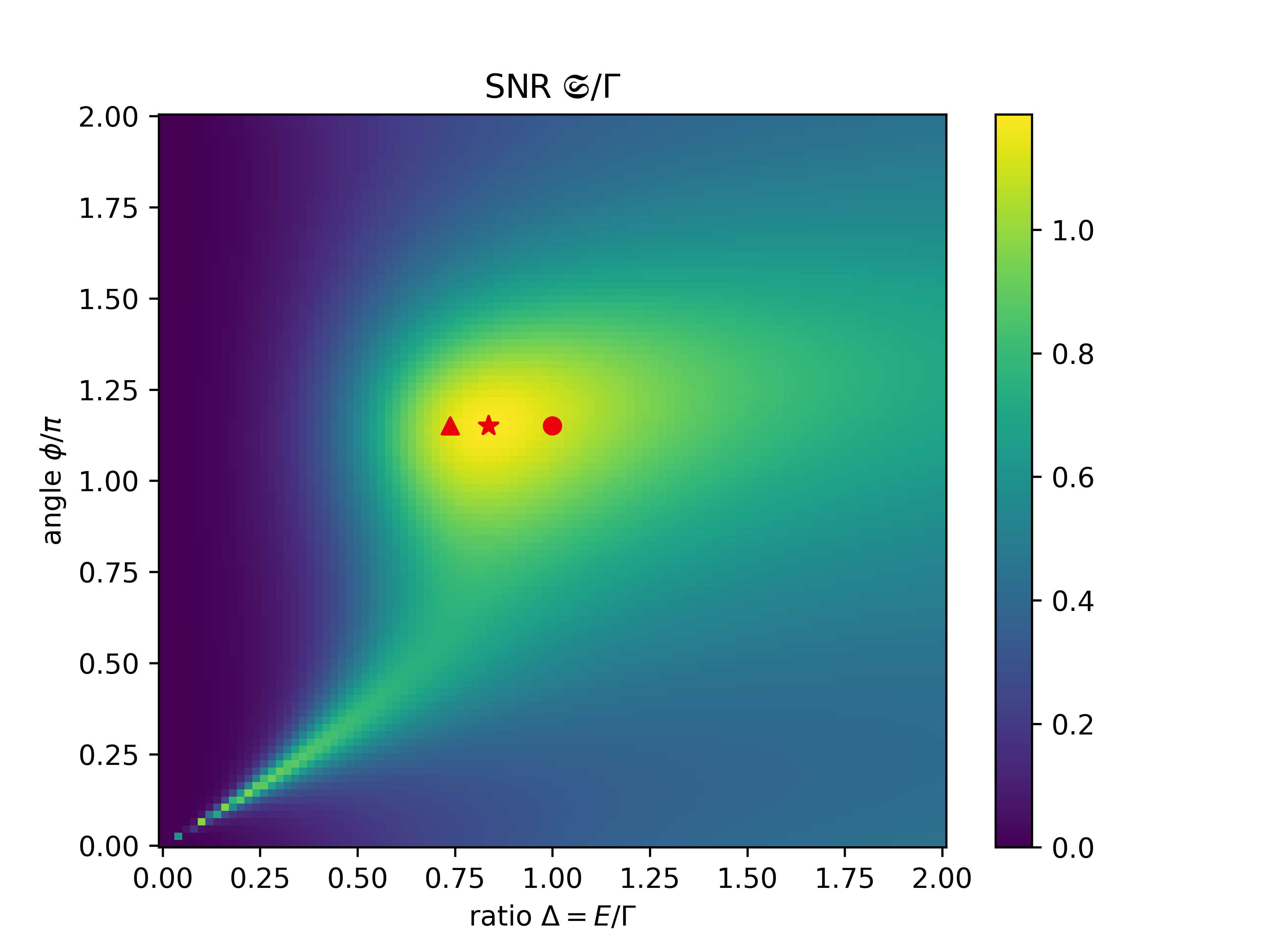}
    \caption{Data from a numerical simulation in Python \cite{Python,Github} of the signal-to-noise ratio of the qubit clockwork and an integrated current, counting each jump in the clockwork, plotted versus the control parameters $E^{(\sidx)}$ and $\phi^{(\sidx)}$.
    The red star marks the parameters that achieve the maximum SNR $\mathmerit^*/\Gamma\approx1.19$ for a single qubit clockwork.
    The red triangle and circle mark the parameters between which the optimal feedback policy for two copies of a qubit clockwork, constructed in \cref{app:go_example_qubit_clockwork}, switches.
    }
    \label{fig:SNR_qubit_clockwork_Python}
\end{figure}

\paragraph{Reference point.} In order to determine the merit of incoherent feedback protocols we compare them against protocols with a \emph{constant feedback policy}.
We call a feedback policy $\Pi$ constant, if it has a trivial memory space $\memoryspace=\{0\}$. 
Then, both its memory-update function $\memoryupdate{}{}{}$ and the parameter-update functions $\memorytorates{\sidx}$ must be constant. Consequently, the control unit chooses one particular set of control parameters at which the clockworks $C_{\sidx}$ are operated at all times $t$, i.e., it implements no feedback.
We deem a general feedback policy useful only if it outperforms any constant feedback policy that has the same level of control as the general one, i.e., that has the same set of allowed parameters $\{\parameterspace{\sidx}\}_\sidx$.

For this comparison to be sensible, it should include the possibility to process the information obtained from the clockworks differently, when implementing a constant versus when implementing a general feedback policy.
How the information is processed is determined by the choice of integrated current, i.e., by the choice of weights $\{\weights{\jidx}\}_\jidx$.
Hence, we compare the SNR of the general feedback policy for some integrated current to the highest signal-to-noise ratio that any integrated current can achieve for a given constant feedback policy, i.e., $\max_{\{\weights{\jidx}\}_\jidx}\mathmerit$. 
In general, such an optimal integrated current for the constant policy differs form an optimal one for the general policy.

\paragraph*{Results.}
In the following, we compare the performance of constant and general feedback policies for classical and quantum clockworks of dimension two. Let us fix an orthonormal basis $\{\ket{0},\ket{1}\}.$
Firstly, we focus on the classical case: 
As discussed in \cref{app:classical_Markovian_sys}, any two-dimensional clockwork $C_\sidx$ satisfying the properties of \textST{} and \textCWI{} that remains classical, i.e., diagonal in the orthonormal basis above, throughout the evolution of the clockwork can be specified by jump operators $J_{C_\sidx,0} = \sqrt{\Gamma_{0}^{(a)}}\ketbra{1}{0}_{C_\sidx}$, $J_{C_\sidx,1} = \sqrt{\Gamma_{1}^{(a)}}\ketbra{0}{1}_{C_\sidx}$ and a trivial Hamiltonian $H_{C_\sidx}=0$. Thus, the only tuneable parameters of a classical clockwork of dimension two are the two jump rates $\Gamma_{0}^{(a)}$ and $\Gamma_{1}^{(a)}$, i.e., the control parameters are $\mathbf{c}=(\Gamma_{0}^{(a)},\Gamma_{1}^{(a)})^{\mathrm{T}}$.
We limit our discussion to feedback policies for which both rates can be tuned within the same nontrivial range $\{0\}\neq\tilde{\parameterspace{a}}\subset\mathbb{R}_{\geq0}$, i.e., a parameter space $\parameterspace{a}=\tilde{\parameterspace{a}}\times\tilde{\parameterspace{a}}$.
For $\numpps$ classical clockworks and feedback policies of this form, we prove that incoherent feedback can never be beneficial in terms of the signal-to-noise ratio.
\begin{restatable}[]{theorem}{classicalSNRbound}
    \label{thm:classical_SNR_bound}
    Consider $\numpps$ classical clockworks of dimension two together with a feedback policy $\Pi$ such that for all $\sidx=1,\ldots,\numpps$ we have $\parameterspace{a}=\tilde{\parameterspace{a}}\times\tilde{\parameterspace{a}}$, for some $\{0\}\neq\tilde{\parameterspace{a}}\subset\mathbb{R}_{\geq0}$, and let the joint system be initialized in a steady state.
    Then, for any integrated current $N(t)$ the signal-to-noise ratio $\mathmerit$ is upper bounded by 
    \begin{equation}
        \mathmerit \leq \max_{m\in\memoryspace}\mleft( \max_{(i_1,\ldots,i_\numpps)\in\{0,1\}^{\numpps}}\sum_{\sidx=1}^{\numpps}\memorytorates{\sidx}_{i_a}(m)\mright).
        \label{eq:maximal_SNR}
    \end{equation}
\end{restatable}
Note, \cref{thm:classical_SNR_bound} holds for arbitrary dimensions, the only restriction being the symmetric parameter space.
The proof relies on the kinetic uncertainty relation in \cite{Terlizzi_Baiesi_KUR} and is given in \cref{app:proof_no_go} together with a proof of the following result.
\begin{restatable}[]{corollary}{nogo}
    \label{cor:No_go}
    Consider $\numpps$ classical clockworks of dimension two together with a feedback policy $\Pi$ as in \cref{thm:classical_SNR_bound}. Then, there exists a constant feedback policy $\Pi'$ and an integrated current $N'(t)$ that satisfy \cref{eq:maximal_SNR} with equality.
\end{restatable}
Therefore, in the case of classical clockworks of dimension two, for which both jump rates $\Gamma_{0}^{(a)}$ and $\Gamma_{1}^{(a)}$ can be tuned in the same range, feedback cannot be beneficial, as for every feedback policy and integrated current one can construct a constant feedback policy and a potentially different integrated current, performing at least equally well or better in terms of the SNR. 
Below we demonstrate that this is surprisingly not the case for quantum clockworks of dimension two, by constructing an explicit example of a clockwork and a feedback policy for $\numpps=2$, that outperforms any constant policy.

Consider a Hamiltonian $H_{C_\sidx} = - \frac{E^{(a)}}{2}\mleft(\ketbra{0}{0} - \ketbra{1}{1}\mright)$, where $E^{(a)}$ denotes the energy difference between the ground state $\ket{0}$ and the excited state $\ket{1}$, as well as a jump operator 
\begin{equation}
    J_{C_a} = \sqrt{\Gamma} \ketbra{\phi^{(a)}}{+}
    = \frac{\sqrt{\Gamma}}{2}
    \begin{pmatrix}
        1 & 1\\
        \mathrm{e}^{i\phi^{(a)}} & \mathrm{e}^{i\phi^{(a)}}
    \end{pmatrix},
\end{equation}
where we defined the states $\ket{+} = \frac{1}{\sqrt{2}}(\ket{0}+\ket{1})$ and $\ket{\phi^{(a)}} = \frac{1}{\sqrt{2}}(\ket{0}+ \mathrm{e}^{i\phi^{(a)}} \ket{1})$. 
When visualizing the state of the clockwork on a Bloch sphere, the ratio between the energy $E^{(a)}$ and the rate $\Gamma^{(a)}$ determines, whether the continuous rotation of the state vector around the equator, induced by the Hamiltonian, or the non-continuous evolution via jumps into the state $\ket{\phi^{(a)}}$ is dominant. Note, however, that due to the backaction there also exists a non-Hermitian part in the continuous evolution in between jumps.
In the following, we will keep the rate $\Gamma$ a fixed parameter and treat both the energy $E^{(\sidx)}$ and the angle $\phi^{(\sidx)}$ as control parameters of the clockwork, i.e., $\mathbf{c}=(E^{(\sidx)}, \phi^{(\sidx)} )^{\mathrm{T}}$.
As we allow for arbitrary values of the control parameters the parameter space is given by $\parameterspace{(\sidx)}=(0,\infty)\times[0,2\pi)$.

Note, that for a single qubit clockwork only a single jump operator and therefore only a single choice of an integrated current $N^{(\sidx)}(t)$ exists - up to an irrelevant multiplicative constant:
the integrated current counting the total number of jumps that have occurred in clockwork $C_\sidx$ up to time $t$.
In \cref{fig:SNR_qubit_clockwork_Python}, we plot the SNR of a qubit clockwork for this integrated current versus the angle $\phi^{(a)}$ in units of $\pi$ and the ratio $\Delta\coloneq E^{(a)}/\Gamma$. 
As indicated by the red star, the maximal SNR ${\mathmerit^*}/{\Gamma} \approx 1.19$ is reached for $\Delta^*\approx0.84$ and $\phi^*/\pi\approx1.15$.
An analytic expression of the SNR and further discussion can be found in \cref{app:SNR_single_qubit_clockwork}.

Let us consider two copies of this clockwork, $C_1$ and $C_2$, together with a constant feedback policy $\Pi$. In \cref{app:proof_constant_feedback_policies_result}, we show that for a fixed constant feedback policy the highest signal-to-noise ratio, achievable by any integrated current, is equal to a sum running over all clockworks $C_\sidx$: Each term in this sum is the highest SNR achievable by any integrated current, when we treat each clockwork $C_\sidx$ as a separate clock.
Thus, the best performing constant feedback policy can be constructed as follows: For each clockwork $C_\sidx$, the policy chooses those control parameters which maximize the highest signal-to-noise ratio, achievable by any integrated current $N^{(\sidx)}(t)$, that counts jumps only in clockwork $C_{\sidx}$. 
From this analysis, it is clear that for two qubit clockworks of the form above, the best performing constant feedback policy $\Pi_{\mathrm{const}}$ operates both clockworks at parameters $E^*$ and $\phi^*$ and achieves an SNR of $2\mathmerit^*\approx2.38\Gamma$, which upper bounds the SNR of any other constant feedback policy for any integrated current.

In contrast to this, we construct a general feedback policy $\Pi_{\mathrm{switch}}$ that achieves a SNR of $\mathmerit \approx 2.59\Gamma$, beating this bound by roughly $9\%$, and that requires only two memory states, in \cref{app:go_example_qubit_clockwork}. 
Informally, this feedback policy keeps the angle constant at $\phi^{(\sidx)}=\phi^*$ for both clockworks $\sidx=1,2$, but switches between two choices of the energy $E^{(\sidx)}$ depending on which of the clockworks experienced a jump last. 
If the last jump occurred in clockwork $C_1$, the energy of clockwork $C_1$ is set to the value corresponding to the red circle in \cref{fig:SNR_qubit_clockwork_Python}, whereas the energy of clockwork $C_2$ is set to the value corresponding to the red triangle and vice versa if the last jump occurred in clockwork $C_2$.
This shows, that even though operating both clockworks at one of these two points constantly would lead to a suboptimal performance in terms of the signal to noise ratio, switching between the two point increases the overall performance of the clock.

\begin{figure}
    \centering
    \includegraphics[width=0.99\linewidth]{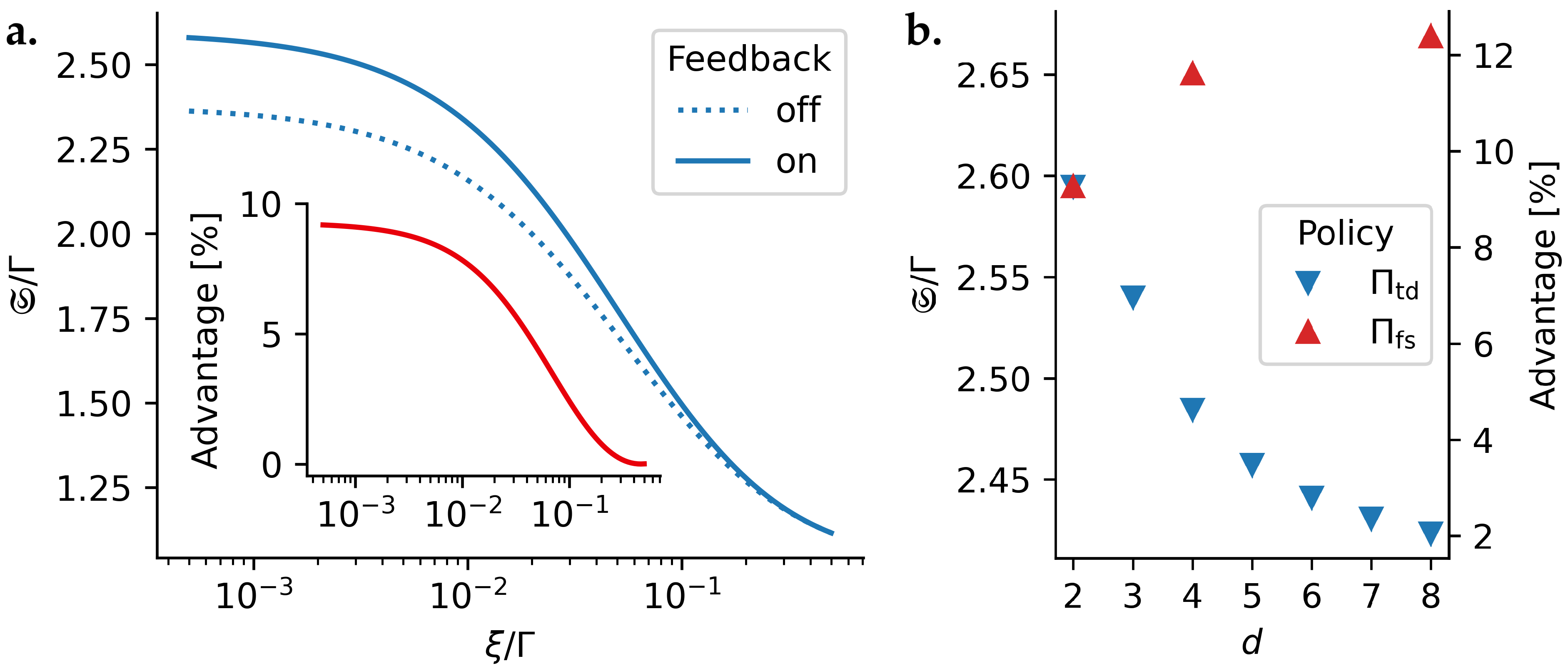}
    \caption{{Results of a numerical simulation in Python \cite{Python,Github}. For all data points, we used an integrated current that weighs jumps from each clockwork equally, independent of the memory state. 
    \textbf{a.} The dotted blue line shows the signal-to-noise-ratio of the best performing constant feedback policy for a fixed dephasing strength $\xi$. The solid blue line shows the performance of a feedback policy, that operates similar to policy $\Pi_{\mathrm{switch}}$, at the same strength $\xi$. The inset shows the relative difference between the two blue curves in percent. \textbf{b.} Signal-to-noise ratio of policies $\Pi_\mathrm{fs}$ and $\Pi_\mathrm{td}$ plotted versus size $d$ of the memory space. The axis on the right-hand side shows the relative improvement over the signal-to-noise ratio of the best performing constant policy $\Pi_\mathrm{
    const}$.}}
    \label{fig:robustness_and_mem_dependence}
\end{figure}

\paragraph{Robustness under dephasing.}
\Cref{cor:No_go} suggest that the coherent evolution of the qubit clockwork is responsible for the advantage of feedback policy $\Pi_\mathrm{switch}$ over any constant policy.
To support this claim, we introduce dephasing noise of strength $\xi$ into our model, by adding jump operators of the form $\sqrt{\xi}\sigma^z_{C_\sidx}$ to the overall system's Lindblad operator.
As the parameters chosen by policies $\Pi_{\mathrm{const}}$ and $\Pi_{\mathrm{switch}}$ may not be optimal at dephasing strengths $\xi\neq0$, we treat $\xi$ as an externally fixed but known parameter and allow policies $\Pi_{\mathrm{const}}$ and $\Pi_{\mathrm{switch}}$ to take this information into account.
The signal-to-noise ratio of these two adapted policies, obtained by separately optimizing over parameters $E^{(\sidx)}$, $\phi^{(\sidx)}$ and $\memorytorates{\sidx}(m)$ for each strength $\xi$, is plotted in \cref{fig:robustness_and_mem_dependence}. One may show that the adapted policy $\Pi_\mathrm{const}$ performs best among all physical, constant feedback policies.
We observe that the advantage of the adapted policy $\Pi_{\mathrm{switch}}$ indeed decreases with increasing dephasing noise, as it disturbs the clockworks' coherent evolution. 
However, the advantage over any constant feedback policy remains stable for strengths $\xi\leq0.01\Gamma$. A more detailed analysis can be found in \cref{app:robustness_against_dephasing}.

\paragraph{Enlarging the memory state space.} We expect that larger memory spaces, storing more information about the sequence of occurred jumps, benefit the highest achievable signal-to-noise ratio. To test this, we study two generalizations of policy $\Pi_{\mathrm{switch}}$ to memory spaces of size $d$:
The memory space of policy $\Pi_{\mathrm{fs}}$ records the entire sequence of the last $\log _2(d)$ jumps, whereas the memory space of policy $\Pi_{\mathrm{td}}$ records how many of the last $d-1$ ticks occurred in clockwork $C_2$.
In \cref{fig:robustness_and_mem_dependence}, we plot the performance of both policies for sizes $2\leq d\leq 8$, where we numerically optimize the policies' parameters $\phi^{(\sidx)}$ and $\memorytorates{\sidx}(m)$ for each $d$ separately, using Python \cite{Python,Github}. 
As expected, the performance of policy $\Pi_\mathrm{fs}$ increases with size $d$, reaching an advantage of roughly 12.5\% at $d=8$, whereas the performance of policy $\Pi_\mathrm{td}$ decreases with increasing size $d$.
This is explained by the fact that after optimization policy $\Pi_\mathrm{td}$ chooses the same parameters as policy $\Pi_\mathrm
{const}$ for almost all memory states. 
For more details on the analysis and a discussion of potential reasons for this behavior, we refer to \cref{app:memory_size_dependence}.

\paragraph{Discussion \& Outlook.} 

We have introduced a framework to describe feedback for ticking clocks, where the clockwork's dynamics are controlled based on classical information about the previous sequence of jumps in different clockworks. In this framework, we differentiate between constant feedback policies, that do not apply feedback and general feedback policies, that do apply feedback and compare their performance in terms of the maximally achievable SNR. For classical clockworks, where the parameters controlling the clockwork's dynamics can be tuned in a symmetric range, we have shown that general feedback policies can never outperform the best constant feedback policy. Whether this also holds for classical clockworks with an arbitrary parameter space remains an interesting question for future research. If this holds, this would provide an extremely interesting separation between classical and quantum clockworks, as we have provided an example of a feedback policy for two qubit clockworks, which outperforms any constant feedback policy for the same system and is robust to dephasing noise. This advantage arises from conditionally switching between two operating points that are each suboptimal on their own. It decreases gradually as dephasing suppresses the clockworks' coherent evolution. Since we evaluated the feedback policies with a single, equally weighted integrated current, while optimizing over all integrated currents for the constant policies, the reported advantages are lower bounds. For ticking clocks, other separations between classical and quantum clockworks have already been shown for example in terms of a different scaling of the clock's quality with the dimension or the resolution~\cite{Meier2023,yang2020ultimatelimittimesignal}. Furthermore, we have provided evidence that this advantage of quantum feedback policies over any constant feedback policy can increase with the amount of stored, classical information about the past sequence of jumps, and that it depends on how this information is organized. Examining this relationship more closely and determining upper bounds on the obtainable advantage for a given memory size, provides an exciting direction for future research. More broadly, our results show that feedback is not automatically useful for autonomous clocks, but can become so for quantum clockworks.

Finally, we have set aside the thermodynamic costs of feedback, such as memory erasure and Hamiltonian switching work. For classical clockworks, where feedback yields no gain, such costs can only make it less attractive. For quantum clockworks, weighing the gain against these costs, as well as against potential degradation of the advantage due to imperfect tick detection and feedback delays, is a natural next step, for which the methods of Refs.~\cite{rosal2025deterministicequationsfeedbackcontrol1,rosal2025deterministicequationsfeedbackcontrol2,rosal2025deterministicequationsfeedbackcontrol3} may prove useful. The feedback control setup that was presented here is classical and we have explored only energy-based feedback, whether coherent feedback or other forms of control such as tuning jump strengths or internal couplings, can increase the advantage further remains an open question.

\section*{Acknowledgements}
 The authors acknowledge fruitful discussions with Tim Ehrensberger, Katrin Holzfeind, Marcus Huber, Florian Meier, Yuri Minoguchi, Mark Mitchison, Monika, Nuriya Nurgalieva, Renato Renner, Ralph Silva, Mischa Woods and Jake Xuereb. We further thank the anonymous reviewer for their useful feedback and suggestions. For the most part, the work presented here are results obtained during J.M.'s MSc ETH master thesis project \cite{Masterarbeit}. J.M. thanks Renato Renner for his support and Marcus Huber for hosting him at the Atominstitut, where the project was carried out. Currently, J.M. is a member of the Department of Mathematical Sciences, University of Copenhagen, Copenhagen, Denmark. P.E. acknowledges funding from the Austrian Federal Ministry of Education, Science, and Research via the Austrian Research Promotion Agency (FFG) through Quantum Austria project 914033. This project is co-funded by the European Union (Quantum Flagship project ASPECTS, Grant Agreement No.\ 101080167). Views and opinions expressed are however those of the authors only and do not necessarily reflect those of the European Union or Research Executive Agency. The European Union can not be held responsible for them.

\begin{figure}[htb!]
\centering
\includegraphics[width=0.4\columnwidth]{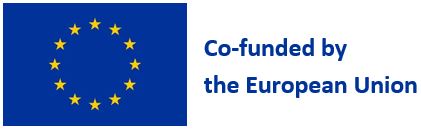}
\end{figure}

\bibliography{references}



\setcounter{secnumdepth}{2}
\onecolumngrid
\section*{Supplemental Material}

\appendix

\begin{appendices}

\section{Classical Markovian systems}
\label{app:classical_Markovian_sys}

We says that a quantum system behaves classical if its density matrix is diagonal in some preferred basis at all times.
In the case of Markovian systems, a system, which is classical in this sense, can equivalently be described as a continuous-time Markov chain.
In this section, we summarize the terminology for continuous-time Markov chains that will be used in the following sections.
Furthermore, for readers unfamiliar with the topic, we briefly outline how to switch between the two equivalent descriptions.\\

Consider a continuous-time Markov chain with states $\cidx\in\calphabet{}$.
The system is described by the distribution $\mathbf{p}(t)=(p_1(t),\ldots,p_{|\calphabet{}|}(t))^\mathrm{T}\in\mathbb{R}^{|\calphabet{}|}$, where $p_{\cidx}(t)$ denotes the probability for the system to occupy state $\cidx$ at time $t$. 
The system's time evolution is determined by a classical Markovian Master equation $\odv{\mathbf{p}}{t}=\mathrm{L}\mathbf{p}(t)$, where the rate matrix~$\mathrm{L}$ has elements
\begin{equation}
    L_{\altcidx\cidx}=\begin{cases}
        R_{\altcidx\cidx} &,\mathrm{if}\;\cidx\neq\altcidx\\
        -\Gamma_{\cidx} &,\mathrm{if}\;\cidx=\altcidx
    \end{cases}
\end{equation} 
where we defined $\Gamma_{\cidx}\coloneq\sum_{\altcidx\in\calphabet{}}R_{\altcidx\cidx}$.
Requiring $0\leq R_{\altcidx\cidx}<\infty$, we can interpret $R_{\altcidx\cidx}$ as the rate, at which a jump of the system from state $\cidx$ to state $\altcidx$ occurs, and $\Gamma_\cidx$ as the total rate of escape from state $\cidx$. 
That is the rate at which jumps of the system from state $\cidx$ to any other state occur.
By convention, we set $R_{\cidx\cidx}=0$ for all $\cidx$.\\

One can equivalently describe the system by replacing the distribution $\mathbf{p}(t)$ with a diagonal density operator $\rho(t)=\mathrm{diag}(p_1(t),\ldots,p_{|\calphabet{}|}(t))$.
Furthermore, replacing the classical with a quantum Markovian Master equation $\odv{\rho}{t}=\Lindblad{}{\rho}{}$ yields the same dynamics, for a Lindblad operator $\Lindblad{}{}{}$ with trivial Hamiltonian $H=0$ and jump operators $J_{\altcidx\cidx}=\sqrt{R_{\altcidx\cidx}}\ketbra{\altcidx}{\cidx}$.
For convenience, we have labelled the jump operators by two indices $\cidx,\altcidx$ instead of a single index $\jidx$.
Note, that the Lindblad operator must necessarily be of this form in order for a density operator describing the system to behave classically, i.e., to remain diagonal at all times $t$ during its time evolution~\cite[Section II.D]{Woods2022}.\\

\section{\label{app:FCS_for_clocks}Analyzing clock performance: exact calculation and upper bounds}

In this section, we briefly summarize how one can calculate exactly or obtain an upper bound on the asymptotic signal-to-noise ratio of a ticking clock, that satisfies \textST{} and \textCWI{} and that is initialized in a steady state.
As discussed in the main text, we can view such a ticking clock as a Markovian quantum system, that is initialized in a steady state, together with an integrated current.
With this view on a ticking clock in mind, we will not mention ticking clocks explicitly in the rest of this section.

\subsection{Exact calculation using full counting statistics}
Given an explicit construction one may use tools from full counting statistics, which we summarize below, in order to calculate the SNR.
For an extensive and pedagogical introduction to this topic refer to \cite{bridging_the_gap}.\\

\paragraph{Vectorization of density matrices and superoperators.}
Let us fix a basis such that we can identify linear operators and matrices.
Consider a linear operator $A=\begin{pmatrix}\mathbf{a}_1\cdots \mathbf{a}_n\end{pmatrix}\in\mathbb{C}^{n\times n}$ consisting of column vectors $\mathbf{a}_1,\ldots,\mathbf{a}_n\in\mathbb{C}^n$.
We define the (basis-dependent) vectorization of $A$ as the vector 
\begin{equation}
     \vecket{A} \coloneq \begin{pmatrix}
    \mathbf{a}_1\\
    \vdots\\
    \mathbf{a}_n\\
    \end{pmatrix}\in\mathbb{C}^{n^2},
\end{equation} 
consisting of the columns of $A$ stacked on top of each other.
Denoting the conjugate transpose of $\vecket{A}$ as $\vecbra{A}$, it is clear that the trace of $A$ can be written as $\trace[A] = \vecbraket{\mathbb{1}}{A}$,
where $\mathbb{1}$ denotes the identity matrix.
Consider a superoperator $\mathcal{L}$ acting by left- and right-multiplication, i.e., $\mathcal{L}[A]=BAC$ for $A,B,C\in\mathbb{C}^{n\times n}$.
Then, the vectorization of $\mathcal{L}[A]$ is given by 
\begin{equation}
    \vecket{\mathcal{L}[A]} = \underbrace{C^\mathrm{T}\otimes B}_{\eqcolon \widehat{\mathcal{L}}} \vecket{A},
\end{equation}
where $C^\mathrm{T}$ denotes the transpose of $C$ and $\otimes$ denotes the Kronecker product.
We refer to the matrix $\widehat{\mathcal{L}}$ as the vectorization of the superoperator $\mathcal{L}$.
By linearity, this definition can be extended to superoperators that act as linear combinations of superoperators of the form above. \\

\paragraph{Calculating the average current and noise.}
Consider a finite-dimensional quantum system governed by a Markovian Master equation with Lindblad operator $\Lindblad{}{}{}=-i[H,\bullet]+\sum_{\jidx\in\jumpalphabet{}}\dissipator{J_\jidx}{}$, with Hamiltonian $H$ and jump operators $J_{\jidx}$, as well as an integrated current $N(t)$. 
Viewing the system's evolution from the equivalent stand point of a stochastic Master equation \cite{bridging_the_gap}, it becomes clear that $N(t)$ the is a random variable that depends on the quantum trajectory the system has taken until time $t$.
We define the average current $F(t)$ and noise $D(t)$ as 
\begin{equation}
    F(t) \coloneq \odv{}{t}\expect{t}{N} \quad \mathrm{and} \quad
    D(t) \coloneq \odv{}{t}\var{t}{N},
\end{equation}
where $\expect{t}{N}$ and $\var{t}{N} = \expect{t}{N^2}-\expect{t}{N}^2$ denote the average and variance of the integrated current $N(t)$ taken over all possible trajectories at time $t$.
Assuming that the system is initialized in a steady state $\rho^{ss}$, one may show that
\begin{equation}
    \label{eq:average_current}
    F(t) = \vecbra{\mathbb{1}}\sum_{\jidx\in\jumpalphabet{}} \weights{\jidx} \mleft(J_{\jidx}^*\otimes J_{\jidx}\mright) \vecket{\rho^{ss}}
\end{equation}
and that $D(t)=D_1(t)-2D_2(t)$, where we defined
\begin{equation}
    \label{eq:noise_term_1}
    D_1(t) = \vecbra{\mathbb{1}}\sum_{\jidx\in\jumpalphabet{}} \weights{\jidx}^2 \mleft(J_{\jidx}^*\otimes J_{\jidx}\mright) \vecket{\rho^{ss}}
\end{equation}
and
\begin{equation}
    \label{eq:noise_term_2}
     D_2(t) = -
      \vecbra{\mathbb{1}} \sum_{\jidx\in\jumpalphabet{}} \weights{\jidx} \mleft(J_{\jidx}^*\otimes J_{\jidx}\mright) \underbrace{\mleft(\int_{0}^t \odif{\tau}\mleft(\mathrm{e}^{\matLindblad{}{\tau}} - \vecketbra{\rho^{ss}}{\mathbb{1}}\mright) \mright)}_{\eqcolon - \matLindblad{}{^+(t)}}
     \sum_{\altjidx\in\jumpalphabet{}} \weights{\altjidx} \mleft(J_{\altjidx}^*\otimes J_{\altjidx}\mright) \vecket{\rho^{ss}},
\end{equation}
where the vectorization of the Lindblad operator is given by 
\begin{equation}
    \matLindblad{}{} = -i\mleft(\mathbb{1}\otimes H - H^{\mathrm{T}}\otimes \mathbb{1}\mright)
    +\sum_{\jidx\in\jumpalphabet{}}\mleft( J_{\jidx}^*\otimes J_{\jidx} -\frac{1}{2}\mathbb{1}\otimes J_{\jidx}^\dagger J_{\jidx}-\frac{1}{2}  J_{\jidx}^\mathrm{T} J_{\jidx}^* \otimes\mathbb{1}\mright).
\end{equation}

From~\cref{eq:average_current,eq:noise_term_1}, we observe that the average current $F(t)$ and the term $D_1(t)$ in the expression for the noise are independent of time $t$.
Note, that for a system not initialized in a steady state this is in general not the case.
For the term $D_2(t)$ in the expression for the noise, only the upper integration limit in $\matLindblad{}{^+(t)}$ depends on $t$.
For large integration times, the integral $\int_0^t\odif{\tau}$ may be replaced by $\int_0^\infty\odif{\tau}$, i.e., $\matLindblad{}{^+(t)}$ may be replaced by $\matLindblad{}{^+}\coloneq \lim_{t\rightarrow\infty}\matLindblad{}{^+}(t)$.
Therefore, $D(t)$ becomes effectively time-independent for large times $t$.
From these observations, one finds that in the limit of large times both the expectation value $\expect{t}{N}$ and the variance $\var{t}{N}$ grow linearly in time $t$ \cite{bridging_the_gap}.\\

Thus, from \cref{eq:SNR} one finds that the asymptotic signal-to-noise ratio is given by $\mathmerit=F^2/D$, where we defined the asymptotic average current $F\coloneq\lim_{t\rightarrow\infty}F(t)$ and noise $D\coloneq\lim_{t\rightarrow\infty}D(t)$. 
Given an explicit expression for $\matLindblad{}{^+}$, one can easily evaluate \crefrange{eq:average_current}{eq:noise_term_2} to determine $F,D$ and $\mathmerit$.\\

\paragraph{The generalized inverse $\matLindblad{}{^+}$.}
For a matrix $A$, that fulfills $\mathrm{rank}(A)=\mathrm{rank}(A^2)$, the group inverse is the unique matrix $A^\#$ that satisfies
\begin{gather}
    A A^\# A = A  \label{eq:group_inverse_prop_1} \\ 
    A^\# A A^\# = A^\# \label{eq:group_inverse_prop_2} \\
    A^\# A = A A^\# \label{eq:group_inverse_prop_3}.
\end{gather}
For a Lindblad operator, whose steady state $\rho^{ss}$ is unique, $\matLindblad{}{^+}=\lim_{t\rightarrow\infty}\matLindblad{}{^+(t)} $ is known to be the group inverse of the matrix $\matLindblad{}{}$.
For convenience, we briefly outline a sketch of the proof.
The fact that the algebraic and geometric multiplicities of the zero eigenvalue of $\matLindblad{}{}$ coincide \cite{PhysRevA.81.062306, PhysRevA.98.042118}, implies $\mathrm{rank}(\matLindblad{}{})=\mathrm{rank}(\matLindblad{}{}^2)$.
Furthermore, the existence of a unique steady state implies that the steady state is attractive and all non-zero eigenvalues have negative real part \cite{PhysRevA.81.062306, Zhang_2024}.
In this case, $\matLindblad{}{^+}$ is well defined and determining $\matLindblad{}{^+}$ from a Jordan decomposition of $\matLindblad{}{^+}$ \cite[Appendix A]{PhysRevLett.113.240406}, allows one to check \crefrange{eq:group_inverse_prop_1}{eq:group_inverse_prop_3} easily.\\

Under the additional assumption of $\matLindblad{}{}$ being diagonalizable, one may show that
\begin{equation}
    \label{eq:Drazin_Moore_penrose}
    \matLindblad{}{^+}=\mleft(\mathbb{1} - \vecketbra{\rho^{ss}}{\mathbb{1}}\mright) \matLindblad{}{^{MP}}\mleft(\mathbb{1} - \vecketbra{\rho^{ss}}{\mathbb{1}}\mright),
\end{equation}
where $\matLindblad{}{^{MP}}$ denotes the Moore-Penrose pseudo-inverse of the matrix $\matLindblad{}{}$, which can be obtained via a singular value decomposition~\cite[Appendix L]{bridging_the_gap}.\\

\subsection{Hyperaccurate currents and uncertainty relations}
\label{app:uncertainty_relations}

In the previous subsection, we have discussed how the asymptotic signal-to-noise ratio can be calculated exactly via the asymptotic average current and noise.
For the proof of \cref{thm:classical_SNR_bound} in \cref{app:proof_no_go}, it suffices to determine an upper bound on the SNR instead.
In literature, such upper bounds on the SNR $\mathmerit(t)$ are known as uncertainty relations, which come in many different flavors.
For example, for classical systems that satisfy local detailed balance and are in a steady state, $\mathmerit(t)$ can be bounded in terms of the system's entropy production \cite{Barato_Seifert_og_TUR}.
As we cannot possibly give a concise overview of the topic of uncertainty relations here, we restrict ourselves to a brief discussion of the kinetic uncertainty relation \cite{Terlizzi_Baiesi_KUR} and the clock uncertainty relation \cite{Prech2025}, for classical Markovian systems as they were discussed in \cref{app:classical_Markovian_sys}.
Note, that some generalizations of uncertainty relations to quantum Markovian systems exist and have gained significant attention recently, see e.g. \cite{Hasegawa2020,Sato2022}. 
We suspect that tight uncertainty relations in the quantum case will be highly useful tools in determining tight bounds on the performance of quantum clockworks in future\cite{Horowitz2020,Hasegawa2020,Hasegawa2021,Moreira2025}.\\

\paragraph{Kinetic uncertainty relation.}
 Let $N(t)=\sum_{\cidx\neq\altcidx}\weights{\altcidx\cidx}N_{\altcidx\cidx}(t)$ be an arbitrary counting observable, where $N_{\altcidx\cidx}(t)$ denotes the elementary integrated current counting the number of jumps from state $\cidx$ to $\altcidx$ that have occurred until time $t$ and where $\weights{\altcidx\cidx}$ is the weight with which such jumps are counted in $N(t)$.
 In \cite{Terlizzi_Baiesi_KUR}, Terlizzi and Baiesi showed that the signal-to-noise ratio of $N(t)$ is upper bounded by
 \begin{equation}
    \label{eq:transient_KUR}
    \mathmerit(t) \leq \frac{\expect{t}{N_{\mathrm{tot}}}}{t} \quad\forall t,
\end{equation}
where $N_{\mathrm{tot}}(t)=\sum_{\cidx\neq\altcidx}N_{\altcidx\cidx}(t)$ is the counting observable that counts every occurring jump with equal weight $\weights{\altcidx\cidx}=1$ for all $\cidx\neq\altcidx$. 
If the system is initialized in a stationary distribution $\mathbf{p}^{ss}$, i.e., a state satisfying $\mathrm{L}\mathbf{p}^{ss}=0$, the expectation value can easily be computed.
Since $\expect{t}{N_{\mathrm{tot}}}$ grows linearly in time for a stationary system, the right-hand side of the inequality above will be equal to the average current $F_{\mathrm{tot}}$.
By switching from a description in terms of a continuous-time Markov chain to a description in terms of a diagonal Markovian quantum system, as described in \cref{app:classical_Markovian_sys}, and using \cref{eq:average_current} one finds 
\begin{equation}
        \mathmerit(t)\leq\mathdynact\quad\forall t,
\end{equation}
where $\mathdynact\coloneq\sum_{\cidx\in\calphabet{}}\Gamma_\cidx p_{\cidx}^{ss}$ is the dynamical activity.\\
    
\paragraph{Clock uncertainty relation.}
Since the right-hand side in the kinetic uncertainty relation does not dependent on the weights of current $N(t)$, one would not expect the inequality to be tight in general.
This is remedied by the clock uncertainty relation derived by Prech et al.\ \cite{Prech2025}.
Consider a continuous-time Markov chain with a unique stationary distribution $\mathbf{p}^{ss}$ and let the initial distribution be equal to that stationary distribution. 
For any integrated current $N(t)$ the asymptotic signal-to-noise ratio is bounded by
\begin{equation}
    \label{eq:CUR}
    \mathmerit\leq \mathrestim^{-1},
\end{equation}
where we defined the mean residual time $\mathrestim\coloneq\sum_{\cidx\in\calphabet{}} \frac{p_{\cidx}^{ss}}{\Gamma_{\cidx}}$.
In contrast to the kinetic uncertainty relation, the clock uncertainty relation can be violated at finite times. 
However, the clock uncertainty relation is tight in the following sense.
For the integrated current 
$N_{BLUE}(t) \coloneq \sum_{\altcidx\neq\cidx}\frac{1}{\Gamma_{\cidx}}{N_{\altcidx,\cidx}(t)}$, \cref{eq:CUR} is satisfied with equality. An integrated current with this property is called a hyperaccurate current \cite{Prech2025}.\\

\section{\label{app:clock_quality} Quantifying clock precision}
In this section, we compare two ways of quantifying the quality of the time estimate that a clock produces, the accuracy $\accuracyN$ and signal-to-noise ratio $\mathmerit$, and argue why $\mathmerit$ is the correct choice for our setting.
For general ticking clocks, Silva et al.~\cite{RNH_ticking_clocks} introduced the accuracy
\begin{equation}
    \accuracyN \coloneq \lim_{t\rightarrow\infty}\frac{\expect{t}{N}}{\var{t}{N}},
\end{equation}
as a figure of merit, where $N(t)$ denotes the random variable that models the outcome of a time measurement on the clock's tick register performed at background time $t$, i.e., a random variable that takes on the register index $n$ with probability $\trace[\rho_{C_n}^{(n)}(t)]$. 
Furthermore, they show that in the case of so-called reset clocks $\accuracyN$ reduces to previously used definitions of clock accuracy \cite{Florian_thesis,yang2020ultimatelimittimesignal,Erker2017,Milburn2020,Woods2021, Meier2023,Xuereb2023}.
As discussed in the main text, in the case of ticking clocks that fulfill \textST{} and \textCWI{} the random variable $N(t)$ is an integrated current.
Throughout the rest of this section, we restrict our attention to this case and additionally always assume the clock to be initialized in a steady state. 
Then, the accuracy is closely related to the asymptotic signal-to-noise ratio 
\begin{equation}
    \mathmerit = \lim_{t\rightarrow\infty} \frac{\mleft(\odv{}{t}\expect{t}{N}\mright)^2}{\var{t}{N}/t},
\end{equation}
which in the context of ticking clocks was first used in~\cite{Prech2025}.
Using the fact that for systems initialized in their steady state the expectation value $\expect{t}{N}$ and variance $\var{t}{N}$ grow linearly with $t$ in the limit of large $t$, we have
\begin{equation}
    \mathmerit = \accuracyN F = \frac{F^2}{D},
\end{equation}
where $F$ and $D$ denote the asymptotic average current and noise, as defined in~\cref{app:FCS_for_clocks}. 
We remark that in the special case of reset clocks, mentioned above, the average frequency $F$ equals the resolution $\nu$ from the introduction~\cite{RNH_ticking_clocks}.
Apart from $F$, other generalizations of the resolution to ticking clocks can be considered.
For an example, that further leads to a notion of accuracy different from the one presented here, we refer to~\cite{Prech2025}. \\

\paragraph{$\mathmerit$ versus $\accuracyN$.}
A figure of merit is used compare clocks $A$ and $B$ with one another and decide which one of the two performs better.
For this comparison to be fair, one should ensure that both clocks, $A$ and $B$, measure the same quantity, i.e., that they both measure time in the same unit.
In the case of ticking clocks, the clock's unit of time is related to the asymptotic average current $F=\lim_{t\rightarrow\infty}\odv{}{t}\expect{t}{N}$. Since the average current measures the average increment by which $N(t)$ increases while one second of background time passes, the inverse $1/F$ determines the time interval (in seconds) which passes on average until the counter $N(t)$ increases by one. 
Therefore, to ensure that clocks $A$ and $B$ measure time in the same unit, we need to ensure that the clock's average frequencies, $F^{(A)}$ and $F^{(B)}$, are equal. 
If they are not, we can imagine to perform one of the following operations.\\

Firstly, one could speed up or slow down the dynamics of clock $B$ by replacing the Lindblad operator $\Lindblad{B}{}{}$ with $\alpha \Lindblad{B}{}{}$.
We refer to this operation as rescaling the dynamics by a factor of $\alpha$.
Secondly, one could post-process the time signal obtained from clock $B$ differently. 
Imagine that it measures time in discrete steps such that $N^{(B)}(t)$ increases by an increment of one for each observed jump.
Then, counting only every second jump, for example, would correspond to replacing $N^{(B)}(t)$ by $N'^{(B)}(t) = \lfloor N^{(B)}(t)/2\rfloor$. 
In the more general case, we replace $N^{(B)}(t)$ by $N'^{(B)}(t) = N^{(B)}(t)/\alpha$, which is equivalent to replacing the weights $\weights{\jidx}$ by $\weights{\jidx}'=\weights{\jidx}/\alpha$, and refer to this operation as rescaling the weights by a factor of $\alpha$.\\

\begin{proposition}
    \label{lem:rescaling_F_and_D}
    Consider a ticking clock satisfying \textST{} and \textCWI{} that is initialized in a steady state $\rho^{ss}$.
    Let $F$ ($D$) and $F'$ ($D'$) denote the average current (noise) before and after a rescaling. Then, we have
    \begin{equation*}
        F' = \alpha F\quad\text{and}\quad D'= \alpha D,
    \end{equation*}
    when rescaling the dynamics by a factor of $\alpha$ and
    \begin{equation*}
        F' = \frac{F}{\alpha}\quad\text{and}\quad D'= \frac{D}{\alpha^2},
    \end{equation*}
    when rescaling the weights by a factor of $\alpha$.
\end{proposition}
\begin{proof}
    Let $\Lindblad{}{}{}$ be the clock's Lindblad operator. From~\cref{app:FCS_for_clocks}, recall
    \begin{equation}
        F = \vecbra{\mathbb{1}}\sum_{\jidx\in\jumpalphabet{}} \weights{\jidx} \mleft(J_{\jidx}^*\otimes J_{\jidx}\mright) \vecket{\rho^{ss}},
    \end{equation}
    and $D=D_1-2 D_2$ with 
    \begin{equation}
        D_1 = \vecbra{\mathbb{1}}\sum_{\jidx\in\jumpalphabet{}} \weights{\jidx}^2 \mleft(J_{\jidx}^*\otimes J_{\jidx}\mright) \vecket{\rho^{ss}},
    \end{equation}
    and
    \begin{equation}
         D_2 = -
          \vecbra{\mathbb{1}} \sum_{\jidx\in\jumpalphabet{}} \weights{\jidx} \mleft(J_{\jidx}^*\otimes J_{\jidx}\mright) \matLindblad{}{^+}     \sum_{\altjidx\in\jumpalphabet{}} \weights{\altjidx} \mleft(J_{\altjidx}^*\otimes J_{\altjidx}\mright) \vecket{\rho^{ss}},
    \end{equation}
    where $\matLindblad{}{^+}$ is the generalized inverse of $\matLindblad{}{}$.
   
    When rescaling the dynamics of $\Lindblad{}{}{}$ by a factor of $\alpha$, it is clear that rescaling the Lindblad operator translates to rescaling the Hamiltonian as $H'=\alpha H$ and the jump operators as $J_{\jidx}'=\sqrt{\alpha}J_{\jidx}$.
    For the generalized inverse we find
    \begin{equation}
        \matLindblad{}{'^+} 
        = \int_{0}^\infty \odif{\tau}\mleft(\mathrm{e}^{\alpha\matLindblad{}{\tau}} - \vecketbra{\rho^{ss}}{\mathbb{1}}\mright) 
        = \frac{1}{\alpha} \int_{0}^\infty \odif{\tau'}\mleft(\mathrm{e}^{\matLindblad{}{\tau'}} - \vecketbra{\rho^{ss}}{\mathbb{1}}\mright) = \frac{1}{\alpha} \matLindblad{}{^+},
    \end{equation}
    as one would expect from a generalized inverse. 
    By linearity, we immediately obtain
    \begin{equation}
        F = \frac{F'}{\alpha},\quad D_1=\frac{D'_1}{\alpha} \quad\text{and}\quad D_2=\frac{D'_2}{\alpha}.
    \end{equation}

    When rescaling the weights by a factor of $\alpha$, the Lindblad operator and jump operators are unchanged and the new weights are given by $\weights{\jidx}'=\frac{\weights{\jidx}}{\alpha}$.
    By linearity, we immediately find
    \begin{equation}
        F = \alpha F',\quad D_1=D'_1 \alpha^2 \quad\text{and}\quad D_2=D'_2 \alpha^2.
    \end{equation} 
\end{proof}

$\,$\\

\begin{corollary}
    \label{cor:rescaling_N_and_S}
    Consider a ticking clock satisfying \textST{} and \textCWI{} that is initialized in a steady state $\rho^{ss}$.
    Let $\mathmerit$ ($\accuracyN$) and $\mathmerit'$ ($\accuracyN'$) denote the signal-to-noise ratio (accuracy) before and after a rescaling.
    Then, we have
    \begin{equation*}
        \mathmerit' = \alpha \mathmerit
        \quad\text{and}\quad 
        \accuracyN' = \accuracyN,
    \end{equation*}
    when rescaling the dynamics by a factor of $\alpha$ and
    \begin{equation*}
        \mathmerit' = \mathmerit
        \quad\text{and}\quad 
        \accuracyN' = \alpha \accuracyN,
    \end{equation*}
    when rescaling the weights by a factor of $\alpha$.
\end{corollary}

To make the comparison between clocks $A$ and $B$ fair, one should, instead of comparing them directly, compare clock $A$ with a version of clock $B$ that is rescaled such that the average frequencies of clock $A$ and of the rescaled version of clock $B$ match, i.e., such that $F^{(A)}=F'^{(B)}$.
This can be achieved for both types of rescaling by choosing the parameter $\alpha$ according to \cref{lem:rescaling_F_and_D}. 
Which type of rescaling to consider depends on which operation is trivial, in the sense that it does not affect the answer to the question that we are examining.\\

If one is interested in constructing the best possible clock for a clockwork space of fixed dimension while satisfying \textST{} and \textCWI{}, such as in \cite{yang2020ultimatelimittimesignal}, rescaling the dynamics can be considered trivial.
As rescaling the dynamics does not change the structure of the Lindblad operator $\Lindblad{}{}{}$, both clock $B$ and its rescaled version will have the same structure.
Therefore, it is always admissible to compare clock $A$ to a rescaled version of $B$ instead of directly comparing clock $A$ to $B$.
According to \cref{cor:rescaling_N_and_S}, this makes the accuracy $\accuracyN$ the natural choice to quantify the clock's performance, since $\accuracyN$ is invariant when rescaling the the dynamics.
Without having to explicitly perform the rescaling of the dynamics of clock $B$, we can immediately compare $\accuracyN^{(A)}$ with $\accuracyN^{(B)}=\accuracyN'^{(B)}$.\\

It is worth keeping in mind, however, that the comparison is made not between clocks $A$ and $B$, but between $A$ and the rescaled version of clock $B$.
This distinction between clock $B$ and its rescaled version becomes important if we take the view of a resource theory.
Since the rescaling might also affect the resources that clock $B$ consumes, rescaling the dynamics cannot be considered trivial.
If we assume that post-processing the time information differently, i.e., implementing an integrated current $N'^{(B)}(t)$ instead of $N^{(B)}(t)$, does not influence the resources consumed by clock $B$, rescaling the weights can be considered trivial.
According to \cref{cor:rescaling_N_and_S}, this makes the signal-to-noise ratio $\mathmerit$ the natural choice to quantify clock performance, since $\mathmerit$ is invariant when rescaling the weights.
Without explicitly having to rescale the weights, we can directly compare $\mathmerit^{(A)}$ to $\mathmerit^{(B)}=\mathmerit'^{(B)}$.\\

When comparing two different feedback policies, as they were introduced in the main text, a similar argument applies.
In general, implementing two different feedback policies leads to two different average frequencies of the time estimate that is produced by the control unit.
Therefore, a rescaling is necessary for a fair comparison between the two policies.
In this case, rescaling the dynamics cannot be considered trivial, since it affects the control parameters that are chosen by the feedback policy.
In particular, rescaling the dynamics by a factor of $\alpha$ can equivalently be viewed as a new choice of the control parameters determining the Hamiltonian's and jump operator's strength (see energy-based and jump-strength-based feedback in \cref{app:assumptions_feedback_policy}). Therefore, rescaling the dynamics transforms a feedback policy into a different feedback policy.
In contrast, rescaling the weights is permitted in this scenario, as it corresponds to a change of how the control unit post-processes the obtained time information to estimate the time.
Consequently, we measure the feedback policy's performance in terms of the signal-to-noise ratio $\mathmerit$.\\

Before we conclude this section, we present an alternative view on the signal-to-noise ratio. 
From~\cref{cor:rescaling_N_and_S}, observe that rescaling the weights by a factor of $\alpha=F$ gives a rescaled current $N'(t)$ with average current $F'= 1$.
Viewing the integrated current as an estimator of the parameter time, such an estimator $N'(t)$ would be called unbiased \cite{Prech2025, estimation_theory_2}.
Since we have $\mathmerit=\mathmerit'=1/D'$, we can interpret $\mathmerit$ of an integrated current $N(t)$ to be the inverse of the noise $D'$ of its equivalent unbiased estimator $N'(t)$.\\

\section{Implicit assumptions in (constant) feedback policies}
\label{app:assumptions_feedback_policy}
In this section, we explicitly state the assumptions that are implicit in \cref{def:feedback_policy} of a feedback policy and the definition of a constant feedback policy.
Firstly, we assume that the experimental control, that the control unit exercises over clockwork $C_{\sidx}$ and its time evolution, can be described by the ability to choose between different sets of parameters specifying the Lindblad operator and hence future time evolution of clockwork $C_\sidx$.\\
\begin{assumption}
    \label{ass:Lindblad_control}
    The Lindblad operator of $C_\sidx$ depends on $K^\sidx$ different parameters $\mathbf{c}\in\parameterspace{\sidx}$, where $\parameterspace{\sidx}\subseteq\mathbb{R}^{K^{\sidx}}$ denotes the set of all allowed combinations of parameters, i.e., we have
    \begin{equation*}
        \Lindblad{C_{\sidx}}{}{\mathbf{c}} = -i[H_{C_{\sidx}}(\mathbf{c}),\bullet] + \sum_{\jidx}\dissipator{J_{C_{\sidx},\jidx}(\mathbf{c})}{}.
    \end{equation*}
\end{assumption}
The exact dependence of the Hamiltonian $H_{C_{\sidx}}(\mathbf{c})$ and jump operators $J_{C_{\sidx},\jidx}(\mathbf{c})$ on the parameters $\mathbf{c}$ and the restrictions on the set $\parameterspace{\sidx}$ determine which kind of control we have over the dynamics of system $C_{\sidx}$ and consequently which kind of feedback we can apply to system $C_\sidx$. \\
Below we give some natural models as examples:
\begin{itemize}
    \item \emph{Energy-based feedback} assumes that we have control over the Hamiltonian's strength, i.e., $H_{C_{\sidx}}(c) = c H_{C_{\sidx}}$, for some fixed Hamiltonian $H_{C_{\sidx}}$ and $\parameterspace{\sidx}\subseteq\mathbb{R}$.
    
\item \emph{Time-based feedback} assumes that we can change the state after a jump occurred, by performing a set of unitary operations $\{U_{C_{\sidx}}(\mathbf{c})\}_{\mathbf{c}\in\parameterspace{\sidx}}$, i.e., 
    $J_{C_{\sidx},\jidx}(\mathbf{c}) = U_{C_{\sidx}}(\mathbf{c})J_{C_{\sidx},\jidx}$, for some fixed jump operator $J_{C_{\sidx},\jidx}$.
    Note, that due to unitarity the backaction term in $\dissipator{J_{C_{\sidx},\jidx}(\mathbf{c})}{}$ remains unchanged.

    \item \emph{Jump-strength-based feedback} assumes that we have control over the rate or measurement strength, i.e., $J_{C_{\sidx},\jidx}(c) = \sqrt{c} J_{C_{\sidx},\jidx}$ for some fixed jump operator $J_{C_{\sidx},\jidx}$ and $\parameterspace{\sidx}\subseteq\mathbb{R}$.
    
    \item\emph{Coupling-based feedback} assumes that system $C_{\sidx}$ may be divided into two subsystems ${D_1}$ and ${D_2}$ with Hamiltonians $H_{D_1}$ and $H_{D_2}$, respectively, which interact coherently according to some interaction Hamiltonian $H_{D_1 D_2}^{(int)}$ whose strength is tunable, i.e., $H_{C_\sidx}(c) = H_{D_1}\otimes \mathbb{1}_{D_2} +  \mathbb{1}_{D_1} \otimes H_{D_2} + c H_{D_1 D_2}^{(int)}$, where $\parameterspace{\sidx}\subseteq\mathbb{R}$.
\end{itemize}
$\,$\\
Secondly, we assume that the time scale on the order of which changes of the control parameters are implemented is much shorter than the time scale on the order of which the evolution of the clockworks take place, which is known as the sudden approximation \cite[page 739ff.]{messiah1981quantum}.\\

\begin{assumption}
    \label{ass:sudden_approximation}
    A change of the control parameters $\mathbf{c}\rightarrow\mathbf{c}'$ is implemented instantaneously such that the states of all systems in the joint clockwork $C$ before and after the change are unaffected, i.e., for a change of parameters at time $T$, we have $\lim_{t\rightarrow T^-}\rho_{C}(t)=\lim_{t\rightarrow T^+}\rho_{C}(t)$.
\end{assumption}
$\,$\\
Thirdly, we need to ensure that the control unit does not rely on additional timing resources, that are not explicitly accounted for by the joint clockwork $C$.
Therefore, we require that the control unit's choice of the parameters $\mathbf{c}$ is fully determined by the control unit's memory storing information about the sequence of jumps that have occurred in the joint clockwork $C$ so far and about previous decisions of the control unit.\\

\begin{assumption}
    \label{ass:deterministic_updates}
    There exists a deterministic and time-independent function $\memorytorates{}$ choosing the current control parameter $\mathbf{c} = \memorytorates{}(m)$ based on the current state of the control unit's memory $m$. 
    Furthermore, there exists a deterministic and time-independent function $\memoryupdate{}{}{}$ such that the state of the control unit's memory is updated from $m$ to $m'=\memoryupdate{m}{\jidx}{}$ immediately after a jump of type $\jidx$ has been observed in $C$.
\end{assumption}
Note, that under these assumptions a change of control parameters can only occur immediately after a jump in the joint clockwork has been observed. This ensures that the control unit by itself has no notion of the passage of time and that any timing resource used by the control unit (similar to a CPU clock for example) must be included in the joint clockwork $C$. 
Together with \cref{ass:sudden_approximation}, this assumption correspond to the detectors observing the jumps having infinite bandwidth. In this work, we do not investigate scenarios, with a time delay between the observation of the jumps and the execution of the feedback due to finite detector bandwidths and changes of the system dynamics occurring over a finite amount of time. 
Lastly, we assume the size of the control unit's memory to be finite, since any realistic physical system can store only a finite amount of information.\\

\begin{assumption}
    The memory state space $\memoryspace$ is finite, i.e., $|\memoryspace|<\infty$.
\end{assumption}

The assumptions presented so far naturally lead us to \cref{def:feedback_policy} of a feedback policy. 
As explained in the main text, we wish to compare our feedback policies to protocols where the control unit remains passive at all times and chooses only one set of control parameters to operate the clockworks at, instead of switching between different sets of parameters.
Every protocol of this form, can be described as a constant feedback policy.\\

\begin{definition}
    For systems $C_{\sidx}$, satisfying~\cref{ass:Lindblad_control} for $\sidx=1,\ldots,{\numpps}$, a constant feedback policy $\Pi$ is a feedback policy with a trivial memory state space $\memoryspace=\{0\}$.
\end{definition}
Since the memory state space of a constant feedback policy contains only a single state, both its memory-update function and its set of parameter-update functions must be constant. Consequently, the control unit chooses one particular set of control parameters, out of the set of all valid parameters (determined by $\{\parameterspace{\sidx}\}_\sidx$), at which the clockworks $C_{\sidx}$ are operated at, at all times $t$.\\

Alternatively, one could also demand the parameter-update function to be constant and allow for a general memory space and general memory-update function.
There is, however, a subtle difference between these two approaches, which has to do with post-processing the time information:
If one introduces a memory system (even in the case of a constant parameter-update function), one automatically increases the number of jump operators since every jump operator $J_{C_\sidx,\jidx}$ is replaced by $|\memoryspace|$ jump operators of the form $J_{C_\sidx,\jidx}\otimes\ketbra{\memoryupdate{m}{\jidx}{}}{m}_M$.
Thereby, one also enlarges the space of possible integrated currents as one has to consider integrated currents $N(t)=\sum_{\sidx=1}^{\numpps}\sum_{\jidx\in\jumpalphabet{\sidx}}\sum_{m\in\memoryspace}\weights{\jidx,m}^{(\sidx)}N^{(\sidx)}_{\jidx,m}(t)$, where $N^{(\sidx)}_{\jidx,m}(t)$ is the elementary integrated current counting the total number of jumps associated to the jump operator $J_{C_\sidx,\jidx}\otimes\ketbra{\memoryupdate{m}{\jidx}{}}{m}_M$, instead of integrated currents $N(t)=\sum_{\sidx=1}^{\numpps}\sum_{\jidx\in\jumpalphabet{\sidx}}\weights{\jidx}^{(\sidx)}N^{(\sidx)}_{\jidx}(t)$, where $N^{(\sidx)}_{\jidx}(t)$ counts the total number of jumps associated to the jump operator $J_{C_\sidx,\jidx}$.
Intuitively, this is the difference between a time estimator, that adds an increment $\weights{\jidx}^{(\sidx)}$ upon observing a jump of type $\jidx$ in clockwork $C_\sidx$, or a time estimator, that additionally takes into account the memory $m$, of the sequence of jumps that have been observed in all clockworks at previous times, and accordingly adds an increment $\weights{\jidx,m}^{(\sidx)}$.
It is apriori not clear whether the latter time estimator can give better estimates than the former.
For indications that taking into account the memory can be beneficial we refer to \cite{Masterarbeit}, however, answering this question lies beyond the scope of this work.
Therefore, we restrict ourselves to constant feedback policies as in \cref{def:feedback_policy}.\\

Before we conclude this section, let us remark that, by definition, one has to consider a non-trivial memory space in order to construct a non-constant feedback policy.
Hence, to make the comparison between non-constant and constant feedback policies fair, we should ensure that an advantage in terms of the SNR, when implementing a non-constant policy, stems from the feedback mechanism and not from the enhanced capabilities of the time estimator. 
For the non-constant feedback policy for two qubit clockworks that we construct in \cref{app:go_example_qubit_clockwork}, we ensure this by restricting ourselves to an integrated current with equal weights $\weights{\jidx,m}^{(\sidx)}=\weights{\jidx}^{(\sidx)}$ for all $m\in\memoryspace$.\\

\section{Previous models of feedback in ticking (quantum) clocks}\label{app:previous_work}
To our knowledge, Barato and Seifert \cite{Barato_Seifert} and Yang et al.\ \cite{yang2019accuracyenhancingprotocolsquantum} have so far been the only ones investigating the potential of protocols in which the dynamics of ticking clocks are externally controlled.
However, only in \cite{yang2019accuracyenhancingprotocolsquantum} this is done with the explicit goal of improving the quality of the clock's estimate of time. 
In this section, we give a short summary of both approaches and for the case of one of the protocols from~\cite{yang2019accuracyenhancingprotocolsquantum} we explicitly show how it can be seen as a particular instance of a feedback policy. In \cite{Barato_Seifert}, Barato and Seifert examine the entropy production in classical, so-called Brownian ticking clocks.
Note, that Barato and Seifert \cite{Barato_Seifert} describe clocks as continuous-time Markov chains.
In the following, we rephrase their model in the language of \cite{RNH_ticking_clocks}.\\

The clockwork of a Brownian clock consists of $\Ridx+1$ states, which form a closed loop. 
From the $\ridx$\textsuperscript{th} state in the chain the system can jump in two directions, either to the $(\ridx+1)$\textsuperscript{st} or the $(\ridx-1)$\textsuperscript{st} state, where the rates are biased such that a jump in the forward direction, i.e., to the $(\ridx+1)$\textsuperscript{st} state, occurs with high likelihood compared to a backwards jump.
Once the clock has reached the end of the chain, a jump in the forward direction will bring it into the zeroth state of the chain.
Correspondingly, a backwards jump brings the clock from the zeroth state to the end of the chain.
Due to the possibility of backwards jumps, Brownian clocks allow for an analysis of the thermodynamic resources they require.
Furthermore, upon such a jump from the $\Ridx$\textsuperscript{th} to the $0$\textsuperscript{th} state, the clock's estimate of time is increased by one unit of time and for a jump in the reverse direction it is decreased by one unit.
This corresponds to an integrated current, which has positive weight for the jumps from the $\Ridx$\textsuperscript{th} to the $0$\textsuperscript{th} state, negative weight for jumps from the $0$\textsuperscript{th} to the $\Ridx$\textsuperscript{th} state and weight zero else.
For such a Brownian clock, Barato and Seifert \cite{Barato_Seifert} consider what they refer to as an external protocol, where the biases between jumps in the forward and jumps in the backwards direction are changed during the operation of the clock.
In particular, they introduce a second external clock, without backwards jumps, and explore a protocol where the biases of the Brownian clock are updated upon a signal from the external clock.
It is clear that this protocol can be neatly phrased as a feedback policy:
In this case, we have two clockworks, one for the Brownian clock and a second for the external clock, as well as a control system, whose memory records jumps only from the external clock. 
Based on its memory state, the control unit sets the biases of jumps in the forward and in the backwards direction, in the clockwork of the Brownian clock, according to some policy.
Barato and Seifert \cite{Barato_Seifert} demonstrate that such an external protocol can achieve a significant reduction of the entropy production in the Brownian clock, while maintaining the same level of quality of its time estimate.
However, whether such an external protocol improves the best achievable time estimate obtained from a Brownian clock is unclear and we will leave this open to further exploration in the future.\\

Before we begin our summary of the approach taken by Yang et al.~\cite{yang2019accuracyenhancingprotocolsquantum}, note that it uses the concept of autonomous clocks, introduced in \cite{alternate_tick_game}. Since an autonomous clock can be seen as a so-called elementary ticking clock with discretized dynamics, we will rephrase the approach of \cite{yang2019accuracyenhancingprotocolsquantum} in terms of ticking clocks instead. For more details on the connection between autonomous clocks and ticking clocks refer to \cite[Section IV.E.]{RNH_ticking_clocks} or to \cite[Section II.B.]{Woods2022}.\\

Yang et al.\ \cite{yang2019accuracyenhancingprotocolsquantum} consider the case of two clocks operated in parallel, which they call the input clock (IC) and the enhancing clock (EC). 
While the input clock controls the dynamics of the enhancing clock, the enhancing clock produces the time signal that is then passed on as the output of the overall system.
In~\cite{yang2019accuracyenhancingprotocolsquantum}, Yang et al.\ introduce four different protocols. 
We will focus on~\cref{alg:accuracy_enhancing_protocol} which we reproduce below, for convenience. 
Although protocol 2 is very similar to~\cref{alg:accuracy_enhancing_protocol}, it would require more general control over the clock's dynamics than we have assumed in our framework. 
Protocol 3 and 4 can be viewed as post-processing of the time signal and therefore correspond to a particular choice of the integrated current.\\

While Yang et al.\ \cite{yang2019accuracyenhancingprotocolsquantum} only require that the input clock produces a time signal with i.i.d.\ interarrival times between consecutive ticks, let us assume for simplicity that the IC is a ticking clock satisfying \textST{}, \textCWI{} and that only allows for jumps of a single type which we label by $\jidx=0$.
Then, the Lindblad operator of the IC is given by 
\begin{equation}
    \Lindblad{C_I}{}{} = -i[H_{C_I},\bullet] + \dissipator{J_{C_I,0}}{},
\end{equation}
where $C_I$ denotes the Hilbert space of the IC's clockwork.
For the enhancing clock Yang et al.\ \cite{yang2019accuracyenhancingprotocolsquantum} assume that the dynamics can be switched between two modes of operation: a dissipative evolution and a unitary evolution in which the clock is silent and does not produce any ticks. 
This could correspond to the situation where a detector monitors the system.
When the detector is turned off, the system $C_E$ will evolve unitarily according to some Hamiltonian $H_{C_E}$.
However, when the detector is turned on, the clockwork experiences a backaction as part of its continuous evolution due to the continuous measurement.
Additionally, the system will experience jumps that can be identified with clicks of the detector.
Let us assume for simplicity, that the EC only allows for one type of jump that we label by $\jidx=1$ and denote $J_{C_E,1}$.
Then, the Lindblad operator of the enhancing clock is given by
\begin{equation}
    \Lindblad{C_E}{}{} = -i[H_{C_E},\bullet] + \dissipator{J_{C_E,1}(c)}{},
\end{equation}
where $C_E$ denotes the Hilbert space of the EC's clockwork and
\begin{equation}
    \label{eq:switching_on_off}
    J_{C_E,1}(c) = \sqrt{c} J_{C_E,1}  = c J_{C_E,1}
\end{equation}
with $c\in\{0,1\}=\parameterspace{EC}$. 
A control parameter $c=0$ corresponds to the unitary evolution and a control parameter $c=1$ to the dissipative evolution of system $C_E$.
Note, that~\cref{eq:switching_on_off} can be classified as jump-strength-based feedback (see \cref{app:assumptions_feedback_policy}).
With the EC and IC as defined above,~\cref{alg:accuracy_enhancing_protocol} consists in switching between the unitary and dissipative dynamics of the EC.
\begin{algorithm}
    \SetAlgorithmName{Protocol}{}{}
    \caption{see \cite[Protocol 1]{yang2019accuracyenhancingprotocolsquantum}}
    \label{alg:accuracy_enhancing_protocol}
    
    \begin{algorithmic}[1]
        \State Initialize the IC and EC and set $c=0$ (unitary dynamics).
        \Loop
        \State Wait for a tick of type $i=0$ from the IC. 
        \State Switch to $c=1$ (dissipative dynamics).
        \State Wait for a tick of type $i=1$ from the EC.
        \State Switch to $c=0$ (unitary dynamics) and output a tick.
        \EndLoop
    \end{algorithmic}
\end{algorithm}

In order to make the connection with our framework, note that $C_I$ and $C_E$ fulfill~\cref{ass:Lindblad_control}.
Let us define a feedback policy $\Pi$ with memory space $\memoryspace=\{0,1\}$, a parameter-update rule $\memorytorates{E}(m) = m$ for the enhancing clock and a memory-update rule
\begin{equation}
    \memoryupdate{m}{\jidx}{} = 
    \begin{cases}
        0, & \text{if } \jidx=1\\
        1, & \text{if } \jidx=0
    \end{cases}
\end{equation}
where the controller's memory is set to $m=0$ if the last observed tick was from the enhancing clock or to $m=1$ if the last observed tick was from the input clock. Using~\cref{eq:Lindblad_evo_conditioned_on_memory}, the Lindblad operator for the joint system $C_I\otimes C_E\otimes M$ is given by
\begin{equation}
    \begin{split}
        \Lindblad{\mathrm{fb}}{}{}
        = & -i[H_{C_I}\otimes\mathbb{1}_{C_E}\otimes\mathbb{1}_{M},\bullet] + \dissipator{J_{C_I,0}\otimes\mathbb{1}_{C_E}\otimes\ketbra{1}{0}_M}{} + \dissipator{J_{C_I,0}\otimes\mathbb{1}_{C_E}\otimes\ketbra{1}{1}_M}{} -\\
        & -i[H_{C_E}\otimes\mathbb{1}_{C_I}\otimes\mathbb{1}_{M},\bullet] + \dissipator{J_{C_E,1}\otimes\mathbb{1}_{C_I}\otimes\ketbra{0}{1}_M}{},
    \end{split}
\end{equation}
which could have been guessed directly from~\cref{alg:accuracy_enhancing_protocol} as well.
Passing on ticks only from the EC as output ticks can be seen as a specific choice of an integrated current $N(t) = w_0N_0(t) + w_1 N_1(t)$ with weight $\weights{0}=0$, where $N_{\jidx}(t)$ is the elementary integrated current, counting the numbers of jumps of type $\jidx$ that have been observed.
Note, that the description in terms of a feedback policy can easily be adapted to protocols where the IC and EC allow for more than a single type of jump.
Furthermore, we remark that Yang et al.\ use the $\epsilon$-inaccuracy \cite[Definition 1]{yang2019accuracyenhancingprotocolsquantum} as a measure of the quality of the output time signal in order to analyze the influence of their protocols.
Since the $\epsilon$-inaccuracy does not give a bound on $\mathmerit$ or $\accuracyN$, we cannot translate their results to aid our analysis.\\

\section{Proof of \cref{thm:classical_SNR_bound} and \cref{cor:No_go}}
\label{app:proof_no_go}
Before we give the proof of \cref{thm:classical_SNR_bound}, let us introduce some notation that will help us to construct the rate matrix that can be identified with a feedback policy in classical clockworks of arbitrary dimension.\\

For a vector $\cidx\in\{0,1\}^\numpps$, we write $\ket{\cidx}_{C_1\cdots C_\numpps}=\ket{\cidx^{(1)},...,\cidx^{(\numpps)}}_{C_1\cdots C_\numpps}=\ket{\cidx^{(1)}}_{C_1}\otimes\cdots\otimes\ket{\cidx^{(\numpps)}}_{C_\numpps}$ to denote the joint state of all clockworks, where clockwork $C_\sidx$ is in state $\cidx^{(\sidx)}\in\{0,1\}$.
Furthermore, by $\mathbb{1}^{(\sidx)}=(0,\ldots,0,1,0,\ldots,0)^{\mathrm{T}}$ we denote the vector that has a $1$ at position $\sidx$ and zeros everywhere else. Let $\oplus$ denote the element-wise addition modulo two. 
It is easy to see that the state associated with the vector $\cidx$ will change to the state associated with $\cidx\oplus\mathbb{1}^{(\sidx)}$, after a jump occurs in the $\sidx$\textsuperscript{th} classical clockwork of dimension two.
From \cref{app:classical_Markovian_sys}, it is clear, that the rate at which such a jump occurs is equal to the rate $\Gamma^{(\sidx)}_{\cidx^{(\sidx)}}$ appearing in jump operator $J_{C_\sidx,{\cidx^{(\sidx)}}}=\sqrt{\Gamma^{(\sidx)}_{\cidx^{(\sidx)}}} \ketbra{\cidx^{(\sidx)}\oplus1}{\cidx^{(\sidx)}}_{C_\sidx}$, where by abuse of notation $\oplus$ denotes addition modulo two. 
Therefore, in the absence of any feedback policy the dynamics of the system are governed by a classical Markovian Master equation with a rate matrix $\mathrm{L}$, given by
\begin{equation}
    \label{eq:L_no_feedback}
    L_{\altcidx\cidx} 
    = \sum_{\sidx=1}^{\numpps} \Gamma^{(\sidx)}_{\cidx^{(\sidx)}} \mleft(\delta_{\altcidx,\cidx\oplus \mathbb{1}^{(\sidx)}}-\delta_{\altcidx,\cidx}\mright)
    = 
    \begin{cases}
        -\Gamma_{\cidx}, & \text{if } \altcidx=\cidx \\   
        \Gamma^{(\sidx)}_{\cidx^{(\sidx)}}, & \text{if } \altcidx = \cidx\oplus\mathbb{1}^{(\sidx)}\\
        0, & \text{else}
    \end{cases}
\end{equation}
where we defined $\Gamma_{\cidx} \coloneq \sum_{\sidx=1}^{\numpps}\Gamma^{(\sidx)}_{\cidx^{(\sidx)}}$ and where $\delta_{\altcidx,\cidx} = \prod_{\sidx=1}^{\numpps} \delta_{\altcidx^{(\sidx)},\cidx^{(\sidx)}}$ with the Kronecker delta $\delta_{\cidx^{(\sidx)},\altcidx^{(\sidx)}}$.
Let us now repeat this analysis for $\numpps$ classical clockworks of dimension two together with a general feedback policy $\Pi=(\memoryspace,\memoryupdate{}{}{}, \memorytorates{})$.
For $\cidx\in\{0,1\}^{\numpps}$ and $m\in\memoryspace$, we write $\ket{\cidx,m}_{C_1\cdots C_\numpps M}=\ket{\cidx^{(1)}}_{C_1}\otimes\cdots\otimes\ket{\cidx^{(\numpps)}}_{C_\numpps}\otimes\ket{m}_M$. Similarly as to before, after a jump occurs in clockwork $C_\sidx$, the joint state of the clockworks and memory will change from the state associated with $(\cidx,m)$ to the state associated with $(\cidx\oplus\mathbb{1}^{(\sidx)},\memoryupdate{m}{\cidx^{(\sidx)}}{})$, where by abuse of notation we denote by $\cidx^{(\sidx)}$ the index in set $\jumpalphabet{\sidx}$ that is associated with the jump $\cidx^{(\sidx)}\rightarrow\cidx^{(\sidx)}\oplus1$ in clockwork $C_\sidx$.
The rate at which this jump occurs is given by $\memorytorates{\sidx}_{\cidx^{(\sidx)}}(m)$, where we use the subscript $\cidx^{(\sidx)}$ to denote the $(\cidx^{(\sidx)})$\textsuperscript{th} element of the tuple $\memorytorates{\sidx}(m)\in\tilde{\parameterspace{\sidx}}\times\tilde{\parameterspace{\sidx}}$.
Therefore, the dynamics of the joint system are governed by a rate matrix $\mathrm{L}^{\Pi}$ with elements
\begin{equation}
    \label{eq:L_pp_with_feedback}
    L^{\Pi}_{(\altcidx,n)(\cidx,m)} 
    = \sum_{\sidx=1}^{\numpps} \memorytorates{\sidx}_{\cidx^{(\sidx)}}(m) \mleft(\delta_{\altcidx,\cidx\oplus \mathbb{1}^{(\sidx)}} \delta_{n,\memoryupdate{m}{\cidx^{(\sidx)}}{}}-\delta_{\altcidx,\cidx}\delta_{n,m}\mright) =
    \begin{cases}
        - \Gamma_{(\cidx,m)} &, \text{ if } (\altcidx,n)=(\cidx,m) \\
        \memorytorates{\sidx}_{\cidx^{(\sidx)}}(m) & ,\text{ if } (\altcidx,n)=(\cidx\oplus \mathbb{1}^{(\sidx)},\memoryupdate{m}{\cidx^{(\sidx)}}{}) \\
        0 &, \text{ else}
    \end{cases}
\end{equation}
where we defined $\Gamma_{(\cidx,m)} = \sum_{\sidx=1}^{\numpps} \memorytorates{\sidx}_{\cidx^{(\sidx)}}(m)$.
Note, that the sum in both $\Gamma_\cidx$ and $\Gamma_{(\cidx,m)}$ contains $\numpps$ terms. 
From this key observation, the following result can be obtained as a corollary to the kinetic uncertainty relation, discussed in \cref{app:uncertainty_relations}.\\

\classicalSNRbound*
\begin{proof}
    Let $\rho^{\mathrm{ss}}_{C_1\cdots C_\numpps M}$ be the steady state in which the joint system is initialized.
    As discussed in \cref{app:classical_Markovian_sys}, the diagonal density operator is associated with a stationary distribution $\mathbf{p}^{\mathrm{ss}}$, whose entries ${p}^{\mathrm{ss}}_{(\cidx,m)}$ are the entries on the diagonal of $\rho^{\mathrm{ss}}_{C_1\cdots C_\numpps M}$. From our discussion above, it is clear that the rate matrix $\mathrm{L}^\Pi$ defined in \cref{eq:L_pp_with_feedback} governs the distribution's time evolution.
    In this setting, we can apply the kinetic uncertainty relation, $\mathmerit(t)\leq\mathdynact$ for all $t$, 
    where the dynamical activity is given by 
    \begin{equation}
        \mathdynact= \sum_{\cidx\in\{0,1\}^{\numpps}, m\in\memoryspace} \Gamma_{(\cidx,m)} p^{\mathrm{ss}}_{(\cidx,m)}.
    \end{equation}
    Therefore, we immediately obtain
    \begin{equation}
        \mathmerit \leq \max_{\cidx\in\{0,1\}^{\numpps}, m\in\memoryspace}\mleft(\Gamma_{(\cidx,m)}\mright) = \max_{m\in\memoryspace}\mleft(\max_{\cidx\in\{0,1\}^{\numpps}}\sum_{\sidx=1}^{\numpps}\memorytorates{\sidx}_{\cidx^{(\sidx)}}(m)\mright).
    \end{equation}
    The maximum is guaranteed to exist since the memory space $\memoryspace$ is finite, by definition.
\end{proof}
$\,$\\
By constructing an explicit constant feedback policy and integrated current that saturate \cref{eq:maximal_SNR} we obtain the following result.\\

\nogo*
\begin{proof}
    For the feedback policy $\Pi=(\memoryspace,\memoryupdate{}{}{},\memorytorates{})$, let $m^*\in\memoryspace$ be the memory state and $(i_1^*,\ldots,i_\numpps^*)$ the bit-string such that $\sum_{a=1}^\numpps\gamma_{i_\sidx^*}^{(\sidx)}(m^*)$ achieves the maximum of the right-hand side of \cref{eq:maximal_SNR}. 
    Consider a constant feedback policy $\Pi'=(\memoryspace',\memoryupdate{}{}{}',\memorytorates{}')$ with trivial memory space $\memoryspace'=\{0\}$, trivial memory update function $\memoryupdate{}{}{}'$ and parameter update functions $\memorytorates{}'^{(\sidx)}(0)=(\memorytorates{\sidx}_{i_\sidx^*}(m^*),\memorytorates{\sidx}_{i_\sidx^*}(m^*))^\mathrm{T}$. 
    Under this feedback policy, for all states $\cidx\in\{0,1\}^{\numpps}$ the rates of escape are equal to $\sum_{a=1}^\numpps\gamma_{i_\sidx^*}^{(\sidx)}(m^*)$, such that $\mathdynact=\mathrestim^{-1}=\sum_{a=1}^\numpps\gamma_{i_\sidx^*}^{(\sidx)}(m^*)$.
    Furthermore, every state $\cidx$ can be reached from every other state $\altcidx\neq\cidx$ in $\numpps$ jumps or less, which implies uniqueness of the stationary distribution \cite[Theorem~2.2~(i)]{CTMC_review}.
    Consider an integrated current that counts the total number of jumps that have occurred in all clockworks of dimension two up to time $t$, i.e., $N'(t)=\sum_{\sidx=1}^\numpps \mleft( N_{0}^{(\sidx)}(t)+N_{1}^{(\sidx)}(t)\mright)$, where the elementary integrated current $N^{(\sidx)}_\jidx(t)$ counts the number of jumps of type $\jidx$ in the $\sidx$\textsuperscript{th} classical clockwork of dimension two.
    Up to an irrelevant multiplicative factor $N'(t)$ is equivalent to the hyperaccurate current $N_{BLUE}(t)$, discussed in \cref{app:uncertainty_relations}, and hence satisfies \cref{eq:maximal_SNR} with equality.
\end{proof}

\begin{remark}
    The results presented in this section can easily be generalized to classical clockworks of arbitrary and different dimension, by replacing the requirement $\cidx^{(\sidx)}\in\{0,1\}$ with the requirement $\cidx^{(\sidx)}\in\{0,1,\ldots,d^{(\sidx)}-1\}$, where $d^{(\sidx)}$ is the dimension of classical clockwork $C_\sidx$, and by replacing addition modulo two by addition modulo $d^{(\sidx)}$.
    Thus, the only limitation of our proof comes from the demand of a symmetric parameter space $\parameterspace{\sidx}=(\tilde{\parameterspace{\sidx}})^{d^{(\sidx)}}$.
    We suspect that this limitation can be lifted by using the clock uncertainty relation to determine a tighter bound on the performance of a general feedback policy than that in \cref{thm:classical_SNR_bound}, but we have been unable to prove this.
\end{remark}

\section{Signal-to-noise ratio of the qubit clockwork}
\label{app:SNR_single_qubit_clockwork}
In this section, we determine the signal-to-noise ratio of a single qubit clockwork. 
Recall the qubit clockwork's Lindblad operator
\begin{equation}
    \Lindblad{}{}{E,\phi}= i\frac{E}{2}[\sigma^z,\bullet] + \Gamma \ketbra{\phi}{+}\bullet \ketbra{+}{\phi} - \frac{1}{2} \{\Gamma\ketbra{+}{+},\bullet\},
\end{equation}
where $\sigma^z$ denotes the Pauli z-matrix and where we neglected the subscripts $C_\sidx$ and the superscripts $(\sidx)$ in the control parameters $E,\phi$ for readability.
As there exists only a single jump operator $J=\sqrt{\Gamma}\ketbra{\phi}{+}$, the integrated current is fully determined by the choice of the weight $\weights{}$.
Since the signal-to-noise ratio is invariant under rescaling of the weights, according to \cref{cor:rescaling_N_and_S}, without loss of generality we set $\weights{}=1$.
Then, using the tools of full counting statistics, we can calculate the average current and noise assuming that the system is initialized in a steady state $\rho^{ss}$.
Let us define the ratio between the energy $E$ and the measurement strength $\Gamma$ by $\Delta=(E/\Gamma)$.
We evaluate \crefrange{eq:average_current}{eq:noise_term_2} using Mathematica \cite{Mathematica,Github}, which gives
\begin{equation}
    \frac{F}{\Gamma} = \frac{2  \Delta^2}{(1 + 4 \Delta^2) - \sqrt{1+4\Delta^2}\mathrm{cos}(\phi-\alpha)},
\end{equation}
where $\alpha = \mathrm{arccos}\mleft(\frac{1}{\sqrt{1+4\Delta^2}}\mright)$, and 
\begin{equation}
    \frac{D}{\Gamma} = \Delta^2 \frac{ (5 - 4 \Delta^2 + 32 \Delta^4) - (1+4\Delta^2)\mathrm{cos}{(2(\phi-\alpha))}
    + 4 \sqrt{1 - 4\Delta^2 + 16\Delta^4}\; \mathrm{cos}{(\phi+\beta)}}{((1 + 4 \Delta^2) - \sqrt{1+4\Delta^2}\mathrm{cos}(\phi-\alpha))^3},
\end{equation}
where $\beta=\arccos\mleft(\frac{-1+4\Delta^2}{\sqrt{1 - 4\Delta^2 + 16 \Delta^4}}\mright)$.
Maximizing $\mathmerit/\Gamma = (F/\Gamma)^2/(D/\Gamma)$ in Mathematica \cite{Mathematica,Github}, we find that the maximum
\begin{equation}
    \frac{\mathmerit^*}{\Gamma} \approx 1.19
\end{equation} is achieved for $\Delta^*\approx0.84$ and $\phi^*\approx1.15\pi$. 
Relating $\Delta^*$ to an energy, we obtain the energy $E^*\coloneq\Delta^*\Gamma$ maximizing the SNR.
In \cref{fig:SNR_Mathematica}, we plot the analytical solution for $\mathmerit/\Gamma$ versus $\Delta$ and $\phi$ in Mathematica \cite{Mathematica,Github}. 
In \cref{fig:SNR_qubit_clockwork_Python}, we plot the signal-to-noise ratio versus $\Delta$ and $\phi$ via a different method: We evaluate each data point by constructing the Hamiltonian and jump operators numerically and solving \crefrange{eq:average_current}{eq:noise_term_2} numerically using Python \cite{Python,Github}.
Here, the parameters $\Delta^*,\phi^*$ that achieve the maximum of the SNR $\mathmerit^*$ are marked by the red star.
Comparing \cref{fig:SNR_qubit_clockwork_Python} and \cref{fig:SNR_Mathematica}, we observe that the numerical simulation agrees nicely with the analytical solution.\\

\begin{figure}
    \centering
    \includegraphics[width=0.5\linewidth]{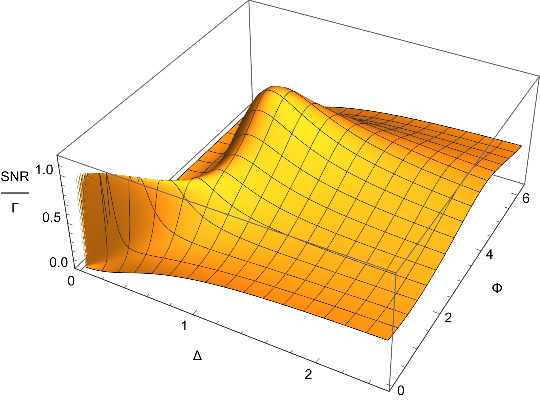}
    \caption{We evaluate and plot the analytical expression for the signal-to-noise ratio $\mathmerit/\Gamma$ using Mathematica \cite{Mathematica,Github}.}
    \label{fig:SNR_Mathematica}
\end{figure}

\section{Constructing the best performing constant feedback policy}\label{app:proof_constant_feedback_policies_result}
In \cref{app:constant_feedback_policies:statement_and_proof}, we show that the best performing constant feedback policy can be constructed by optimizing the performance of each clockwork, separately. 
In \cref{app:constant_feedback_policies:technical_proposition}, we provide a technical result that is used in this proof.

\subsection{A technical intermediate step \label{app:constant_feedback_policies:technical_proposition}}
Consider a Markovian quantum system consisting of two independent Markovian subsystems, $C_1$ and $C_2$.
For an integrated current $N(t)$ that is the sum of an integrated current $N^{(1)}(t)$ and $N^{(2)}(t)$
for systems $C_1$ and $C_2$, respectively, one would expect the expectation value (variance) of $N(t)$ to be the sum of the expectation values (variances) of $N^{(1)}(t)$ and $N^{(2)}(t)$.
The proposition below asserts that this is true. 
\begin{proposition}
    \label{prop:additivity_F_D_independent_subsys}
    Consider a Markovian quantum system consisting of two independent Markovian subsystems, $C_1$ and $C_2$, governed by a Lindblad operator $\Lindblad{C_1 C_2}{}{} = (\mathcal{L}_{C_1}\otimes\mathcal{I}_{C_2})[\bullet]+(\mathcal{I}_{C_1}\otimes\mathcal{L}_{C_2})[\bullet]$, where $\mathcal{I}_{C_\sidx}$ denotes the identity channel and $\Lindblad{C_a}{}{} = -i[H_{C_\sidx},\bullet]+\sum_{\jidx\in\jumpalphabet{\sidx}}\dissipator{J_{C_\sidx,\jidx}}{}$ the Lindblad operator of subsystem $C_\sidx$. 
    Consider an integrated current $N(t) = \sum_{\jidx_1\in\jumpalphabet{1}}\weights{\jidx_1}^{(1)} N_{\jidx_1}^{(1)}(t) + \sum_{\jidx_2\in\jumpalphabet{2}}\weights{\jidx_2}^{(2)} N_{\jidx_2}^{(2)}(t)$, where $N_{\jidx}^{(\sidx)}(t)$ denotes the elementary integrated current counting jumps of type $\jidx$ in subsystem $C_\sidx$.
    If the total system is initialized in a steady state $\rho^{ss}_{C_1 C_2}=\rho^{(1)}_{C_1}\otimes\rho^{(2)}_{C_2}$, the average current $F$ and noise $D$ of the integrated current $N(t)$ are given by 
    \begin{equation*}
        F = F^{(1)} + F^{(2)} \quad\text{and}\quad D = D^{(1)} + D^{(2)},
    \end{equation*}
    where $F^{(\sidx)}$ and $D^{(\sidx)}$ are the average integrated current and noise for a Markovian quantum system with Lindblad operator $\Lindblad{C_\sidx}{}{}$ and an integrated current $N^{(\sidx)}(t) = \sum_{\jidx\in\jumpalphabet{\sidx}}\weights{\jidx}^{(\sidx)}N_{\jidx}^{(\sidx)}(t)$.    
\end{proposition}
\begin{proof}
    Since $\rho_{C_1C_1}^{ss}$ is a steady state of the joint system, we have
    \begin{equation}
        0 = \trace_{C_1}\mleft[\Lindblad{C_1 C_2}{\rho_{C_1 C_2}^{ss}}{}\mright] = \mleft(\trace_{C_1}\circ\Lindblad{C_1}{\rho^{(1)}_{C_1}}{}\mright)\otimes\rho_{C_2}^{(2)} + \trace_{C_1}\mleft[\rho_{C_1}^{(1)}\mright]\otimes\Lindblad{C_2}{\rho_{C_2}^{(2)}}{}
        = \Lindblad{C_2}{\rho_{C_2}^{(2)}}{},
    \end{equation}
    and similarly $\Lindblad{C_1}{\rho_{C_1}^{(1)}}{}=0$, 
    where we used the normalization of $\rho_{C_\sidx}^{(\sidx)}$ and the fact that $\trace_{C_\sidx}\circ\Lindblad{C_\sidx}{}{}=0$ due to the cyclicity of the trace. 
    Define channels $\mathcal{E}_{C_\sidx}[\bullet]\coloneq\sum_{\jidx\in\jumpalphabet{\sidx}}\weights{\jidx}^{(\sidx)} J_{C_\sidx,\jidx}\bullet J_{C_\sidx,\jidx}^\dagger$ and $\mathcal{F}_{C_\sidx}[\bullet]\coloneq\sum_{\jidx\in\jumpalphabet{\sidx}}\mleft(\weights{\jidx}^{(\sidx)}\mright)^2 J_{C_\sidx,\jidx}\bullet J_{C_\sidx,\jidx}^\dagger$.
    Using \crefrange{eq:average_current}{eq:noise_term_2}, we find 
    \begin{equation}
        \label{eq:average_current_C_1_C_2}
        \begin{split}
        F &= \trace\mleft[(\mathcal{E}_{C_1}\otimes\mathcal{I}_{C_2}+\mathcal{I}_{C_1}\otimes\mathcal{E}_{C_2})[\rho^{ss}_{C_1 C_2}]\mright]
        \\ &= \trace\mleft[\mathcal{E}_{C_1}\mleft[\rho^{(1)}_{C_1}\mright]\mright]+\trace\mleft[\mathcal{E}_{C_2}\mleft[\rho^{(2)}_{C_2}\mright]\mright] 
        \\ &=F^{(1)} + F^{(2)}
        \end{split}
    \end{equation}
    where we identified the average currents $F^{(\sidx)}=\vecbra{\mathbb{1}}\sum_{\jidx\in\jumpalphabet{\sidx}}\weights{\jidx}^{(\sidx)}  \mleft(J_{C_\sidx,\jidx}^*\otimes J_{C_\sidx,\jidx}\mright)\vecket{\rho^{(\sidx)}_{C_\sidx}}$ of a Markovian system, governed by the Lindblad operator $\Lindblad{C_\sidx}{}{}$ and initialized in the steady state $\rho^{(\sidx)}_{C_\sidx}$, and of an integrated current determined by the weights $\weights{\jidx}^{(\sidx)}$ for $\jidx\in\jumpalphabet{\sidx}$.
    Analogously, we find 
    \begin{equation}
        \begin{split}
            D_1
            & = \trace\mleft[(\mathcal{F}_{C_1}\otimes\mathcal{I}_{C_2}+\mathcal{I}_{C_1}\otimes\mathcal{F}_{C_2})[\rho^{ss}_{C_1 C_2}]\mright] \\
            & = \trace\mleft[\mathcal{F}_{C_1}\mleft[\rho^{(1)}_{C_1}\mright]\mright]+\trace\mleft[
            \mathcal{F}_{C_2}\mleft[\rho^{(2)}_{C_2}\mright]\mright] \\
            & = D_1^{(1)} + D_1^{(2)},
        \end{split}
    \end{equation}
    where we identified the first term $D_1^{(\sidx)}=\vecbra{\mathbb{1}}\sum_{\jidx\in\jumpalphabet{\sidx}}\mleft(\weights{\jidx}^{(\sidx)}\mright)^2  \mleft(J_{C_\sidx,\jidx}^*\otimes J_{C_\sidx,\jidx}\mright)\vecket{\rho^{(\sidx)}_{C_\sidx}}$ of the noise $D^{(\sidx)}$.
    According to \cref{eq:noise_term_2}, for the second term $D_2$ of the noise, we have 
    \begin{equation}
        \label{eq:noise_term_2_ind_subsys}
        \begin{split}
         D_2 = - \int_0^\infty \odif{\tau} \biggl( &
         \trace\mleft[ \mleft(              \mathcal{E}_{C_1}\otimes\mathcal{I}_{C_2}+
            \mathcal{I}_{C_1}\otimes\mathcal{E}_{C_2}\mright)\circ \mathrm{e}^{\mathcal{L}_{C_1 C_2}\tau}\circ\mleft( \mathcal{E}_{C_1}\otimes\mathcal{I}_{C_2}+
            \mathcal{I}_{C_1}\otimes\mathcal{E}_{C_2}\mright)[\rho^{ss}_{C_1 C_2}] 
         \mright] \\
         &\quad - \trace\mleft[
            \mleft( \mathcal{E}_{C_1}\otimes\mathcal{I}_{C_2}+
            \mathcal{I}_{C_1}\otimes\mathcal{E}_{C_2}\mright)[\rho^{ss}_{C_1 C_2}] 
         \mright]^2
         \biggr).
        \end{split}
    \end{equation}
    Utilizing the fact that $\mathrm{e}^{\mathcal{L}_{C_1 C_2}\tau}[\bullet] = \left(\mathrm{e}^{\mathcal{L}_{C_1}\tau}\otimes\mathrm{e}^{\mathcal{L}_{C_2}\tau}\right)[\bullet]$, for an operator of the form $\Lindblad{C_1 C_2}{}{}=(\mathcal{L}_{C_1}\otimes\mathcal{I}_{C_2})[\bullet]+(\mathcal{I}_{C_1}\otimes\mathcal{L}_{C_2})[\bullet]$, we find
    \begin{equation}
        \label{eq:noise_part_2_C_1_C_2}
        \begin{split}
        &\trace\mleft[ \mleft(              \mathcal{E}_{C_1}\otimes\mathcal{I}_{C_2}+
            \mathcal{I}_{C_1}\otimes\mathcal{E}_{C_2}\mright)\circ \mathrm{e}^{\mathcal{L}_{C_1 C_2}\tau}\circ\mleft( \mathcal{E}_{C_1}\otimes\mathcal{I}_{C_2}+
            \mathcal{I}_{C_1}\otimes\mathcal{E}_{C_2}\mright)[\rho^{ss}_{C_1 C_2}] 
         \mright] = \\
         &\quad \begin{split}
         =&\trace\Bigl[ \Bigl(\mleft(\mathcal{E}_{C_1}\circ\mathrm{e}^{\mathcal{L}_{C_1}\tau}\circ\mathcal{E}_{C_1}\mright)\otimes \mathrm{e}^{\mathcal{L}_{C_2}\tau}  +  \mathrm{e}^{\mathcal{L}_{C_1}\tau} \otimes \mleft(\mathcal{E}_{C_2}\circ\mathrm{e}^{\mathcal{L}_{C_2}\tau}\circ\mathcal{E}_{C_2}\mright) + \\ & \quad+\mleft(\mathcal{E}_{C_1}\circ\mathrm{e}^{\mathcal{L}_{C_1}\tau}\mright)\otimes \mleft(\mathrm{e}^{\mathcal{L}_{C_2}\tau}\circ\mathcal{E}_{C_2}\mright) + \mleft(\mathrm{e}^{\mathcal{L}_{C_1}\tau}\circ\mathcal{E}_{C_1}\mright)\otimes \mleft(\mathcal{E}_{C_2}\circ\mathrm{e}^{\mathcal{L}_{C_2}\tau}\mright)\Bigr)[\rho_{C_1 C_2}^{ss}]\Bigr]
         \end{split}\\
         &\quad 
         =\trace\mleft[ \mleft(\mathcal{E}_{C_1}\circ\mathrm{e}^{\mathcal{L}_{C_1}\tau}\circ\mathcal{E}_{C_1}\mright)\mleft[\rho_{C_1}^{(1)}\mright] \mright] + 
         \trace\mleft[ \mleft(\mathcal{E}_{C_2}\circ\mathrm{e}^{\mathcal{L}_{C_2}\tau}\circ\mathcal{E}_{C_2}\mright)\mleft[\rho_{C_2}^{(2)}\mright] \mright]+2\trace\mleft[\mathcal{E}_{C_1}\mleft[\rho_{C_1}^{(1)}\mright] \mright] \trace\mleft[\mathcal{E}_{C_2}\mleft[\rho_{C_2}^{(2)}\mright] \mright],
         \end{split}
    \end{equation}
    where we used that $\mleft(\trace_{C_\sidx}\circ\mathrm{e}^{\mathcal{L}_{C_\sidx}\tau}\mright)[\bullet]=\trace_{C_\sidx}[\bullet]$, since $\trace_{C_\sidx}\circ\Lindblad{C_\sidx}{}{}=0$, and that $\mathrm{e}^{\mathcal{L}_{C_\sidx}\tau}\mleft[\rho^{(\sidx)}_{C_\sidx}\mright]=\rho^{(\sidx)}_{C_\sidx}$, since $\rho^{(\sidx)}_{C_\sidx}$ is a steady state of $\Lindblad{C_\sidx}{}{}$, in the second equality.
    Plugging \cref{eq:average_current_C_1_C_2,eq:noise_part_2_C_1_C_2} into \cref{eq:noise_term_2_ind_subsys}, we find
    \begin{equation}
        \begin{split}
         D_2 = - \int_0^\infty \odif{\tau} \Bigl( &
         \trace\mleft[ \mleft(\mathcal{E}_{C_1}\circ \mathrm{e}^{\mathcal{L}_{C_1}\tau}\circ\mathcal{E}_{C_1}\mright)\mleft[\rho^{(1)}_{C_1}\mright] 
         \mright] - \trace\mleft[\mathcal{E}_{C_1}\mleft[\rho^{(1)}_{C_1}\mright] 
         \mright]^2 + \\
         &\quad + \trace\mleft[ \mleft(\mathcal{E}_{C_2}\circ \mathrm{e}^{\mathcal{L}_{C_2}\tau}\circ\mathcal{E}_{C_2}\mright)\mleft[\rho^{(2)}_{C_2}\mright] 
         \mright] - \trace\mleft[\mathcal{E}_{C_2}\mleft[\rho^{(2)}_{C_2}\mright] 
         \mright]^2
         \Bigr).
        \end{split}
    \end{equation}
    Identifying $D_2^{(\sidx)}= -
      \vecbra{\mathbb{1}} \sum_{\jidx\in\jumpalphabet{\sidx}} \weights{\jidx}^{(\sidx)} \mleft(J_{C_\sidx,\jidx}^*\otimes J_{C_\sidx,\jidx}\mright) \mleft(\int_{0}^t \odif{\tau}\mleft(\mathrm{e}^{\matLindblad{C_\sidx}{\tau}} - \vecketbra{\rho^{(\sidx)}_{C_\sidx}}{\mathbb{1}}\mright) \mright)
     \sum_{\altjidx\in\jumpalphabet{\sidx}} \weights{\altjidx}^{(\sidx)} \mleft(J_{C_\sidx,\altjidx}^*\otimes J_{C_\sidx,\altjidx}\mright) \vecket{\rho^{(\sidx)}_{C_\sidx}}$ as the second term of the noise $D^{(\sidx)}$, together with the fact that $D=D_1-2D_2$, concludes the proof.
\end{proof}

\subsection{The construction \label{app:constant_feedback_policies:statement_and_proof}}

Let us begin by fixing a constant feedback policy. From the following result, we will be able to determine the highest signal-to-noise ratio that any integrated current can achieve for this policy.
\begin{lemma}
    \label{lem:limits_constant_feedback_policy}
    Consider a constant feedback policy $\Pi$ for clockworks $C_{\sidx}$ for $\sidx=1,\ldots,\numpps$, that are initialized in a steady state $\rho^{ss}_{C_1\cdots C_{\numpps}}=\rho^{(1)}_{C_1}\otimes\cdots\otimes\rho^{(\numpps)}_{C_{\numpps}}$.
    For an integrated current $N(t)$ the signal-to-noise ratio satisfies
    \begin{equation*}
        \mathmerit \leq \sum_{\sidx=1}^{\numpps} \mathmerit^{(\sidx)},
    \end{equation*}
    where $\mathmerit^{(\sidx)}$ is the signal-to-noise ratio of the integrated current $N^{(\sidx)}(t)$, which is obtained from the decomposition $N(t) = \sum_{\sidx=1}^{\numpps}N^{(\sidx)}(t)$, where $N^{(\sidx)}(t)$ only counts jumps in clockwork $C_\sidx$.
    Furthermore, there exists a choice $v^{(\sidx)}\in\mathbb{R}$ for $\sidx=1,\ldots,\numpps$ such that the integrated current $N'(t)=\sum_{\sidx=1}^{\numpps}v^{(\sidx)}N^{(\sidx)}(t)$ saturates this inequality.
\end{lemma}
The proof is given at the end of this section.
Phrased differently, for a fixed constant feedback policy, the SNR of the joint system is at most equal to the sum of the SNRs when treating each clockwork $C_\sidx$ separately, i.e., when treating $C_\sidx$ as a ticking clock satisfying \textST{} and \textCWI{} together with an integrated current $N^{(\sidx)}(t)=\sum_{\jidx\in\jumpalphabet{a}}\weights{\jidx}^{(\sidx)}N_\jidx^{(\sidx)}(t)$ and dynamics governed by the Lindblad operator $\Lindblad{C_\sidx}{}{}=-i[H_{C_{\sidx}}(\memorytorates{\sidx}(0)),\bullet]+\sum_{\jidx\in\jumpalphabet{\sidx}}\dissipator{J_{C_{\sidx},\jidx}(\memorytorates{\sidx}(0))}{}$.
Moreover, by appropriate post-processing of each $N^{(\sidx)}(t)$ (by rescaling its weights by a factor of $1/v^{(\sidx)}$), one can always obtain an integrated current $N'(t)$ that achieves equality.
From this, by taking a maximum over all possible choices of integrated currents, i.e., over all possible sets of weights $\mleft\{\weights{\jidx}^{(\sidx)}\mright\}_{\jidx,\sidx}$, we get the following result, as claimed in the main text: The highest achievable signal-to-noise ratio of the joint system is equal to the sum over $\sidx=1,\ldots,\numpps$ of the highest achievable signal-to-noise ratio in clockwork $C_\sidx$, when treated as a separate clock.\\

Up to this point we have considered the highest signal-to-noise ratio, that can be achieved by any integrated current, for a fixed feedback policy.
Let us now compare this highest achievable signal-to-noise ratio for different constant feedback policies. We observe that the highest achievable SNR of the joint system will be maximal if for each clockwork $C_\sidx$ the feedback policy chooses the control parameters $\memorytorates{\sidx}(0)$ such that the highest achievable SNR in $C_\sidx$, treated as an individual clock, is maximal. Such a constant feedback policy, together with an optimal integrated current and optimal post-processing, performs best in terms of the SNR and consequently is the point of reference to compare general feedback policies with. Let us now prove \cref{lem:limits_constant_feedback_policy}.\\

\begin{proof}
    In the case of $\numpps=1$, the result is trivial.
    In the case of $\numpps>2$, the result can be obtained by iteratively applying the result for $\numpps=2$, as we will describe later.
    Therefore, it suffices to consider $\numpps=2$:\\
    
    Consider an integrated current $N(t)=N^{(1)}(t)+ N^{(2)}(t)$ with $N^{(\sidx)}(t)=\sum_{\jidx\in\jumpalphabet{a}}\weights{\jidx}^{(\sidx)}N_\jidx^{(\sidx)}(t)$ for $\sidx=1,2$. 
    We define a new integrated current $N'(t)=N^{(1)}(t)+ rN^{(2)}(t)=N^{(1)}(t)+ N'^{(2)}(t)$, where we defined $N'^{(2)}(t)=rN^{(2)}(t)$ for $r\in\mathbb{R}$. 
    In the language of \cref{app:clock_quality}, $N'^{(2)}(t)$, is obtained from $N^{(2)}(t)$ by a rescaling of the weights by a factor of $1/r$.\\

    Applying \cref{prop:additivity_F_D_independent_subsys}, we have $F'=F^{(1)}+F'^{(2)}$ and $D'=D^{(1)}+D'^{(2)}$, where $F'$, $F^{(1)}$, $F'^{(2)}$ is the average current and $D'$, $D^{(1)}$, $D'^{(2)}$ is the noise of the integrated current $N'(t)$, $N^{(1)}(t)$ , $N^{(2)}(t)$, respectively.
    Furthermore, since $N'^{(2)}(t)$ is obtained by a rescaling of the weights by a factor of $1/r$, we can apply \cref{lem:rescaling_F_and_D}, to find $F'^{(2)}=r F^{(2)}$ and $D'^{(2)}=r^2 D^{(2)}$, where $F^{(2)}$ and $D^{(2)}$ are the average current and noise of the integrated current $N^{(2)}(t)$.\\

    Putting all of this together, for the SNR of the integrated current $N'(t)$, we have 
    \begin{equation}
        \label{eq:SNR_with_r}
        \mathmerit'=\frac{\mleft(F^{(1)} + r F^{(2)}\mright)^2}{D^{(1)}+r^2 D^{(2)}}.
    \end{equation}
    In order to determine the maximum value of this expression over $r$, we solve $0=\pdv{}{r} \mathmerit'$. We find two solutions 
    \begin{equation}
        r_{\min} = - \frac{F^{(1)}}{F^{(2)}} \quad\text{and}\quad r_{\max} =\frac{F^{(2)}D^{(1)}}{F^{(1)}D^{(2)}}.
    \end{equation}
    Plugging these solutions into \cref{eq:SNR_with_r}, yields
    \begin{equation}
        \label{eq:SNR_r_evaluated}
        \mathmerit'|_{r=r_{\min}} = 0 \quad \text{and}\quad         \mathmerit'|_{r=r_{\max}} = \frac{\mleft(F^{(1)}\mright)^2}{D^{(1)}} + \frac{\mleft(F^{(2)}\mright)^2}{D^{(2)}} = \mathmerit^{(1)} + \mathmerit^{(2)}.
    \end{equation}
    In order to determine whether the points are local maxima or minima, we evaluate the second derivative of $\mathmerit'$ with respect to $r$ at $r_{\min}$ and $r_{\max}$, which gives
    \begin{equation}
        \mleft.\pdv[order=2]{\mathmerit'}{r}\mright|_{r=r_{\min}}
        =\frac{2 \mleft(F^{(2)}\mright)^4}{\mleft(F^{(2)}\mright)^2 D^{(1)} + \mleft(F^{(1)}\mright)^2 D^{(2)}} \geq 0
    \end{equation}
    and
    \begin{equation}
        \mleft.\pdv[order=2]{\mathmerit'}{r}\mright|_{r=r_{\max}}
        =-\frac{2 \mleft(F^{(1)}\mright)^4\mleft(D^{(2)}\mright)^2}{\mleft(D^{(1)}\mright)^2\mleft(\mleft(F^{(2)}\mright)^2 D^{(1)} + \mleft(F^{(1)}\mright)^2 D^{(2)}\mright)} \leq 0,
    \end{equation}
    where we used the fact that $D^{(1)}$ and $ D^{(2)}$ are non-negative.
    From the observation that $\mathmerit'$ is continuous in $r$ and asymptotically approaches $\mathmerit^{(2)}$ for $r\rightarrow\pm\infty$, we find that $r=r_{\min}$ is a global minimum and $r=r_{\max}$ is a global maximum of the SNR $\mathmerit'$. 
    Therefore, we immediately obtain $\mathmerit\leq\mathmerit^{(1)} + \mathmerit^{(2)}$ using $\mathmerit'|_{r=1}=\mathmerit$ and \cref{eq:SNR_r_evaluated}.
    Additionally, we find that the upper bound is saturated by the integrated current $N^{(1)}(t)+r_{\max} N^{(2)}(t)$.\\
    
    In the case of $\numpps>2$, the inequality is obtained by iteratively shrinking and dividing the joint system into two subsystems, for which we apply the inequality found for $\numpps=2$: In the $\ridx$\textsuperscript{th} step of the iteration we partition the system into two subsystems, $C_{\ridx}$ and $C_{\ridx+1}\cdots C_{\numpps}$, and consider integrated currents, $N^{(\ridx)}(t)$ and $\sum_{\sidx=\ridx+1}^{\numpps}N^{(\sidx)}(t)$. 
    Combining the inequalities from each step, will give the desired inequality.
    Similarly, we can iteratively prove achievability:
    In the $\ridx$\textsuperscript{th} step of the iteration, we consider subsystems, $C_{1}\cdots C_{\ridx}$ and $C_{\ridx+1}$, and two integrated currents, $\tilde{N}(t)$ and $N^{(\ridx+1)}(t)$, where $\tilde{N}(t)$ is the integrated current found in the previous round. 
    By determining the constant $r_\mathrm{max}^{(\ridx+1)}$, we obtain an integrated current $\tilde{N}(t)+r_\mathrm{max}^{(\ridx+1)}N^{(\ridx+1)}(t)$ with SNR $\sum_{\sidx=1}^{\ridx+1}\mathmerit^{(\sidx)}$.
\end{proof}

\section{A feedback policy $\Pi_{\mathrm{switch}}$ for two qubit clockworks that outperforms any constant feedback policy}
\label{app:go_example_qubit_clockwork}

In this section, we consider two qubit clockworks. For this system, we demonstrate that there exists a feedback policy $\Pi_{\mathrm{switch}}$ which outperforms any constant feedback policy. 
Note, that for classical clockworks of dimension two such a feedback policy cannot exist in the case of a symmetric parameter space, as discussed in the main text and shown in \cref{app:proof_no_go}.\\

For the qubit clockwork we assumed the rate $\Gamma$ to be fixed and treated the energy difference $E^{(\sidx)}$, between ground state $\ket{0}_{C_\sidx}$ and exited state $\ket{1}_{C_\sidx}$, and the angle $\phi^{(\sidx)}$, appearing in the jump operator, as tunable parameters. 
The feedback policy that we construct in the following will keep the angle $\phi^{(\sidx)}$ fixed at $\phi^*$ for both $\sidx=1,2$ and only change the energy of each clockwork.
Then, the Hamiltonian of system $C_\sidx$ can be written as $H_{C_\sidx}(c)= - \frac{c E^*}{2}\mleft(\ketbra{0}{0}_{C_\sidx} - \ketbra{1}{1}_{C_\sidx}\mright)=-\frac{cE^*}{2}\sigma^z_{C_\sidx}$, where $\sigma^z$ denotes the Pauli $z$-matrix and where the control parameter $c$ determines the ratio between the energy difference set by the control unit and the energy difference $E^*$, which is optimal in the case of a single qubit clockwork. 
Furthermore, the jump operator of system $C_\sidx$ is given by $J_{C_a,\sidx} = \sqrt{\Gamma} \ketbra{\phi^*}{+}_{C_\sidx}$, where the subscript $\sidx$ indicates that we refer to a jump occurring in system $C_\sidx$ as a jump of type $\sidx$.
For the definition of $\phi^*$ and $E^*$ we refer to \cref{app:SNR_single_qubit_clockwork}.\\

Intuitively, the feedback policy $\Pi_{\mathrm{switch}}$ switches the energy of each of the two qubit clockworks between two values, depending on whether clockwork one or clockwork two has produced a tick last.
In order to formalize this, consider a memory space $\memoryspace=\{1,2\}$.
If the last jump occurred in clockwork $C_1$, the state $m$ of the memory is given by $m=1$.
If the last jump occurred in clockwork $C_2$, we have $m=2$.
This gives a memory update function 
 \begin{equation}
    \label{eq:memory_update_go_example}
     \memoryupdate{m}{\sidx}{} = \sidx.
\end{equation}
After a jump occurs, the feedback policy may switch the energy of the two clockworks.
For this purpose, we introduce the parameters $\alpha_1,\alpha_2\in\mathbb{R}$. 
Given, that the last jump originated from clockwork $C_1$, i.e., for the memory in state $m=1$, the feedback policy sets the energy such that the Hamiltonian of clockwork $C_1$ becomes $H_{C_1}(\alpha_1)=-\frac{\alpha_1 E^*}{2}\sigma_{C_1}^z$ and such that the Hamiltonian of clockwork $C_2$ becomes $H_{C_2}(\alpha_2)=-\frac{\alpha_2 E^*}{2}\sigma_{C_2}^z$. 
If the last tick came from clockwork $C_2$, i.e., for the memory in state $m=2$, the feedback policy reverses the energy of the two clockworks. 
More precisely, the Hamiltonian of clockwork one is set to $H_{C_1}(\alpha_2)=-\frac{\alpha_2 E^*}{2}\sigma_{C_1}^z$ and the Hamiltonian of clockwork two is set to $H_{C_2}(\alpha_1)=-\frac{\alpha_1 E^*}{2}\sigma_{C_2}^z$.
This is achieved by parameter-update functions  
\begin{equation}
    \label{eq:paramter_update_function_go_example}
    \memorytorates{\sidx}(\sidxprime) = \begin{cases}
        \alpha_1 & , \text{ if } \sidx=\sidxprime\\
        \alpha_2 & , \text{ if } \sidx\neq\sidxprime
    \end{cases}
\end{equation}
for $\sidx,\sidxprime=1,2$.
Remark, for $\alpha_1=\alpha_2=1$ we would recover the case of no feedback, since the energy of both clockworks would be $E^*$ regardless of the memory state.
According to \cref{eq:Lindblad_evo_conditioned_on_memory}, the Lindblad operator of the joint system is given by 

\begin{equation}
    \begin{split}
        \label{eq:Lindblad_op_go_example}
        \Lindblad{\mathrm{fb}}{}{} = 
\sum_{m=1}^{2} \bigl(&-i[H_{C_1}(\memorytorates{1}(m))\otimes\mathbb{1}_{C_2}\otimes\ketbra{m}{m}_M,\bullet] + \dissipator{J_{C_1,1}\otimes\mathbb{1}_{C_{2}}\otimes\ketbra{1}{m}_M}{} - \\
        & -i[\mathbb{1}_{C_1}\otimes H_{C_2}(\memorytorates{2}(m))\otimes\ketbra{m}{m}_M,\bullet] + \dissipator{\mathbb{1}_{C_{1}}\otimes J_{C_2,2}\otimes\ketbra{2}{m}_M}{}
        \bigr),
    \end{split}
\end{equation}
where we used that fact that the memory is updated to $m=1$, whenever a jump in the first clockwork, occurs and to $m=2$, whenever a jump in the second clockwork occurs.\\

Let us now analyze the quality of a time estimate, that is obtained from the overall clockwork under feedback policy $\Pi_{\mathrm{switch}}$.
We consider the integrated current that simply counts the total number of jumps that have occurred until time $t$, i.e., $N(t)=N^{(1,1)}(t)+N^{(1,2)}(t)+N^{(2,1)}(t)+N^{(2,2)}(t)$, where $N^{(\sidx,m)}(t)$ is the elementary integrated current counting the number of jumps from clockwork $C_\sidx$, given that the memory state before the jump was $m$.
In \cref{fig:qubit_clockwork_feedback_landscape}, we plot $\mathmerit/\Gamma$ for this integrated current versus the ratios $\alpha_1$ and $\alpha_2$.
\begin{figure}
    \centering
    \includegraphics[width=0.65\linewidth]{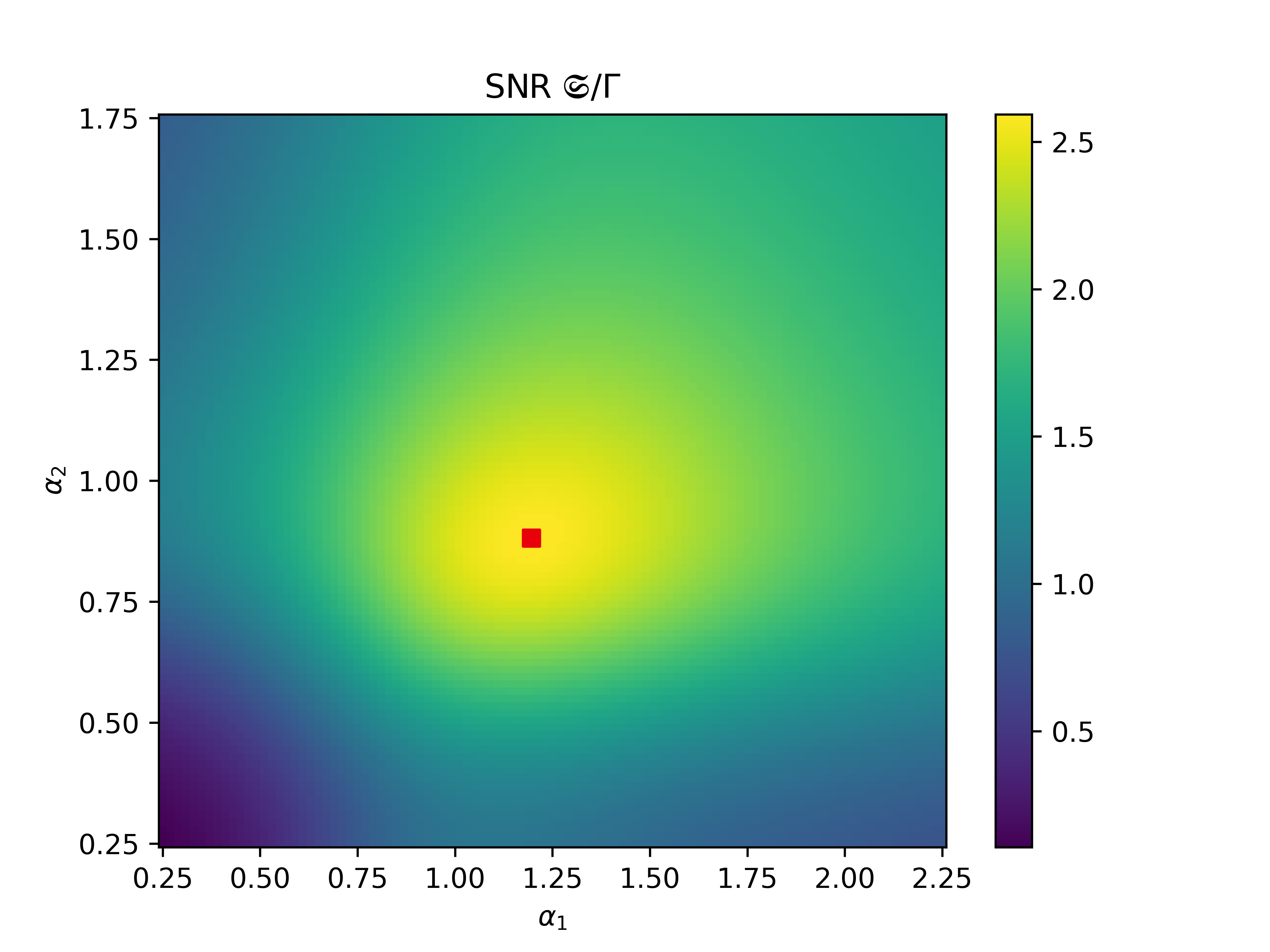}
    \caption{Data from a numerical simulation in Python \cite{Python,Github}, showing the signal-to-noise ratio of an integrated current counting every jump from two qubit clockworks, under energy-based feedback, with equal weights.
    The parameters $\alpha_1$ and $\alpha_2$ determine at which energy the first clockwork is operated at, after a tick from qubit clockwork one or two occurred, and vice versa for the energy of qubit clockwork two.
    The red square marks the parameters maximizing the SNR.}
    \label{fig:qubit_clockwork_feedback_landscape}
\end{figure}
Unsurprisingly, at the point $\alpha_1=\alpha_2=1$ we achieve an SNR $\mathmerit\approx 2.38\Gamma$.
This point corresponds to no feedback and both clockworks being operated at $E^*$ (and $\phi^*$), which would be optimal parameters when considering only a single qubit clockwork.
However, for two qubit clockworks with feedback this point is suboptimal in terms of the SNR.
The maximum of the SNR $\mathmerit\approx 2.59\Gamma$ is reached at the point $\alpha_1\approx1.20$ and $\alpha_2\approx0.88$, indicated by the red square.
Therefore, the feedback policy that we have constructed in this section yields a SNR that is roughly 9\% higher compared to the highest SNR, achievable by any integrated current for any constant feedback policy. 
Consequently, it constitutes an example of a general feedback policy that outperforms any constant feedback policy.\\

Let us briefly outline why one might find this result surprising:
Under the feedback policy with optimal parameters $\alpha_1$ and $\alpha_2$, the energy of clockwork $C_1$ is increased to $1.2E^*$, given that a jump in clockwork $C_1$ has occurred last, speeding up the continuous evolution of clockwork $C_1$. The evolution of clockwork $C_2$ on the other hand is slowed down by setting the energy of the second clockwork to $0.88E^*$. 
If the last jump occurred in clockwork $C_2$, the settings for the two clockworks are reversed.
In \cref{fig:SNR_qubit_clockwork_Python} we plotted the signal-to-noise ratio of a single qubit clockwork and an integrated current counting the total number of ticks observed from this clockwork.
In this figure we have marked the point corresponding to $1.2E^*$ (and $\phi^*$) with a red circle and the point corresponding to $0.88E^*$ (and $\phi^*$) with a red triangle.
From \cref{lem:limits_constant_feedback_policy}, it is clear that a constant feedback policy, which operates the first clockwork at the point marked by the red circle and the second clockwork at the point marked by the red triangle, the best achievable signal-to-noise ratio for any integrated current is upper bounded by $2.38\Gamma$.
However, by switching the dynamics of the clockworks between the points marked by the red circle and triangle, conditioned on which clock ticked last, the feedback policy that we have constructed in this section achieves a signal-to-noise ratio of $2.59\Gamma>2.38\Gamma$. 
Phrased differently, we achieve an increase in performance by (conditionally) switching between two suboptimal points of operation.\\

Our result sheds new light onto a question raised in \cite{RNH_ticking_clocks}: Can it be beneficial in terms of a ticking clock's performance to switch between different dynamics of the clockwork?
Our results suggest, that in the case of multiple quantum clockworks this is indeed the case.
Note, however, that the notion of switching between different dynamics in \cite{RNH_ticking_clocks} is slightly different from what we have presented. 
While by Silva et al.\ \cite{RNH_ticking_clocks} propose to deterministically switch between different dynamics depending on the current state of the tick register, i.e., the overall clock's estimate of time, our incoherent feedback protocol switches between different dynamics depending on the observed sequence of ticks, which is fundamentally probabilistic.

{
\section{Robustness of the advantage under dephasing noise}
\label{app:robustness_against_dephasing}
In \cref{app:go_example_qubit_clockwork}, we gave the example of an energy-based feedback policy $\Pi_{\mathrm{switch}}$ for two qubit clockworks $C_1$ and $C_2$ that outperforms any constant feedback policy.
In this section, we study how robust this advantage is when adding dephasing noise.
For this purpose, we extend the Lindblad operator by adding jump operators of the form $\sqrt{\xi}\sigma^z$ that model qubit dephasing, where $\xi$ denotes the dephasing strength.
Starting from \cref{eq:Lindblad_op_go_example}, this yields a new Lindblad operator
\begin{equation}
    \begin{split}
        \label{eq:Lindblad_op_go_example_under_dephasing}
        \Lindblad{\mathrm{fb}}{}{} = 
\sum_{m=1}^{2} \biggl(&-i[H_{C_1}(\memorytorates{1}(m))\otimes\mathbb{1}_{C_2}\otimes\ketbra{m}{m}_M,\bullet] + \dissipator{J_{C_1,1}\otimes\mathbb{1}_{C_{2}}\otimes\ketbra{1}{m}_M}{} + \dissipator{\sqrt{\xi}\sigma^z_{C_{1}}\otimes\mathbb{1}_{C_{2}}\otimes\ketbra{m}{m}_{M}}{} - \\
        & -i[\mathbb{1}_{C_1}\otimes H_{C_2}(\memorytorates{2}(m))\otimes\ketbra{m}{m}_M,\bullet] + \dissipator{\mathbb{1}_{C_{1}}\otimes J_{C_2,2}\otimes\ketbra{2}{m}_M}{}
        +  \dissipator{\mathbb{1}_{C_{1}}\otimes\sqrt{\xi}\sigma^z_{C_{2}}\otimes\ketbra{m}{m}_M}{}\biggr),
    \end{split}
\end{equation}
with jump operators $J_{C_\sidx,\sidx}=\sqrt{\Gamma}\ketbra{\phi^{(a)}}{+}_{C_{\sidx}}$.
We compare the performance of this policy to the performance of the feedback policy $\Pi_{\mathrm{const}}$
when adding dephasing noise of strength $\xi$.
Recall that in the absence of dephasing noise, the feedback policy $\Pi_\mathrm{const}$ performs the best among all constant feedback policies.
It can be viewed as two  copies of single qubit clockwork $C_{\sidx}$ that are operated at parameters $E^{(\sidx)}=E^*$ and $\phi^{(\sidx)}=\phi^*$. 
When adding dephasing noise to a single qubit clockwork $C_\sidx$, this yields a new Lindblad operator 
\begin{equation}
    \label{eq:Lindblad_single_qubit_clockwork_dephasing}
    \Lindblad{}{}{} = -\frac{E^{(\sidx)}}{2}\sigma^z_{C_\sidx} + \dissipator{J_{C_\sidx,\sidx}}{} + \dissipator{\sqrt{\xi}\sigma^z_{C_\sidx}}{}.
\end{equation}

Note, that we may restrict our analysis to integrated currents that do not count any jumps associated to the newly added dephasing operators, as this information is inaccessible to the control unit.
Therefore, in this section, we consider the same integrated currents as we did in the case without dephasing:
For a single qubit clockwork with dephasing there exists only one possible integrated current - up to an irrelevant rescaling of the weights - counting the jumps associated to the jump operator $J_{C_\sidx,\sidx}$.
For two qubit clockworks with dephasing under the feedback policy from the previous section, we consider 
$N(t)=N^{(1,1)}(t)+N^{(1,2)}(t)+N^{(2,1)}(t)+N^{(2,2)}(t)$, where $N^{(\sidx,m)}(t)$ is the elementary integrated current counting the number of jumps associated to jump operator $J_{C_\sidx,\sidx}$, given that the memory state before the jump was $m$.

In \cref{fig:non_opt_robustness}, we plot the the signal-to-noise ratio of both $\Pi_{\mathrm{switch}}$ and $\Pi_{\mathrm{const}}$ under dephasing as the solid and dotted blue lines, versus strength $\xi$.
We observe that the performance of both policies decreases with increasing strength $\xi$, up to the point $\xi\approx0.4\Gamma$ after which a moderate increase of performance can be seen. 
Policy $\Pi_\mathrm{switch}$ maintains an advantage of more than $5\%$ over policy $\Pi_{\mathrm{const}}$ for $\xi\lesssim 0.05\Gamma$, demonstrating the robustness of the advantage.
For $\xi\gtrsim \Gamma$ any advantage vanishes.

\begin{figure}[htb!]
    \centering
    \includegraphics[width=0.5\linewidth]{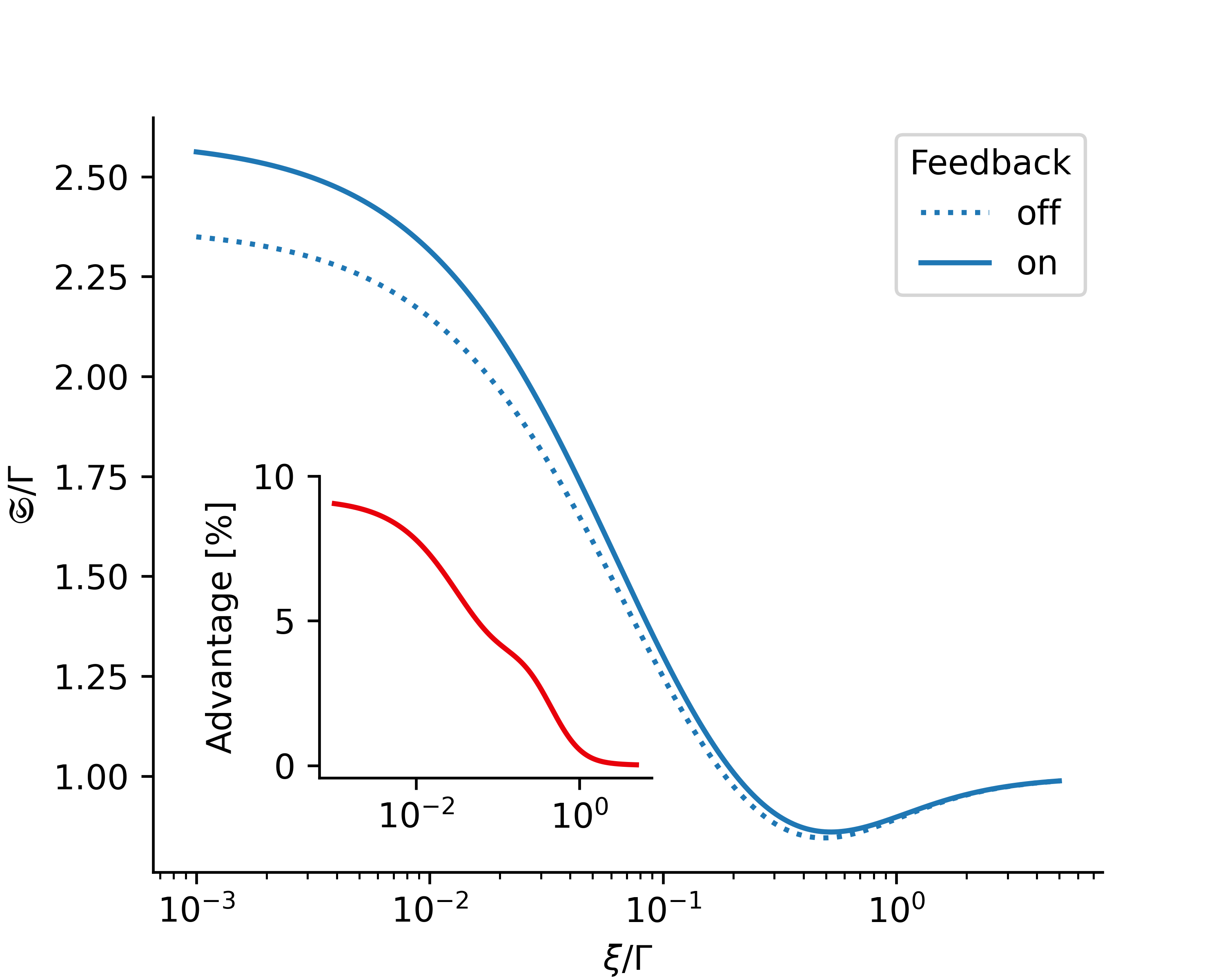}
    \caption{Results of a numerical simulation in Python \cite{Python,Github}, where the solid and dotted blue lines show the signal-to-noise ratio of policies $\Pi_{\mathrm{switch}}$ and $\Pi_{\mathrm{const}}$, when adding dephasing noise of strength $\xi$. In both cases, we consider an integrated current that counts all jumps associated with jump operators $J_{C_\sidx,\sidx}$ equally, but does not count jumps associated with the dephasing operators $\sigma^z_{C_\sidx}$. The inset shows the relative difference between the two blue curves in percent.}
    \label{fig:non_opt_robustness}
\end{figure}

If the dephasing strength is known or can be reasonably well estimated, this information may be taken into account.
In the following, we study this situation, which has been discussed briefly in the main text, by adapting the policy $\Pi_{\mathrm{switch}}$ to choose different parameters $\phi^{(\sidx)}$ and $\memorytorates{\sidx}{}(m)$, depending on the fixed, external parameter $\xi$.

Repeating the analysis outlined in the main text, as a point of reference, we determine the best performing constant feedback, that does not have access to information about the jumps associated with the dephasing operators $\sigma^z_{C_\sidx}$, since any policy with access to this information must be deemed unphysical.
Its signal-to-noise ratio is equal to twice the highest signal-to-noise ratio that is achievable with a single qubit clockwork whose time evolution is governed by \cref{eq:Lindblad_single_qubit_clockwork_dephasing}.
Similar to the analysis described in \cref{app:SNR_single_qubit_clockwork}, we numerically optimize the parameters $E^{(\sidx)}$ and $\phi^{(\sidx)}$ for such a qubit clockwork for fixed values $\xi$ using Python \cite{Python,Github}
and plot the resulting upper bound on the performance of any constant feedback policy as the blue dotted line in \cref{fig:robustness_and_mem_dependence}.
 
For comparison, using Python \cite{Python,Github}, we numerically optimize the parameters $\phi^{(\sidx)}$ and $\memorytorates{\sidx}(m)$ for the non-constant policy defined by \cref{eq:Lindblad_op_go_example_under_dephasing} for the same values $\xi$ . 
As we expect an optimal policy to treat both clockworks symmetrically, we restrict our optimization to parameters $\phi^{(1)}=\phi^{(2)}$, $\memorytorates{1}(1)=\memorytorates{2}(2)$ and $\memorytorates{1}(2)=\memorytorates{2}(1)$.
The resulting signal-to-noise ratio is plotted as the solid blue line in \cref{fig:robustness_and_mem_dependence}. 
We observe that the advantage over any constant feedback policy, at the same dephasing strength, persists but decreases with increasing strength $\xi$.
For values $\xi\leq0.01\Gamma$, the advantage appears stable, whereas it is substantially reduced for values larger than this and vanishes around $\xi\approx0.5\Gamma$.\\

Recent experiments provide illustrative comparisons across several platforms. To relate the dephasing strengths considered above to exemplary experiments, we identify $\Gamma$ with the energy-relaxation rate $1/T_1$ and $\xi$ with the pure-dephasing rate $1/T_\varphi$, such that $\xi/\Gamma = T_1/T_\varphi$. In state-of-the-art superconducting transmon qubits, Hahn-echo coherence times now reach or exceed $T_1$ at the millisecond scale~\cite{Tuokkola2025,Bland2025}, corresponding to an order of magnitude estimate of $\xi/\Gamma\sim10^{-1}$. For a superconducting fluxonium qubit \cite{Watanabe2025} the following numbers were reported, $T_1=51\mu\mathrm{s}$, $T_2=32.9\mu\mathrm{s}$ and a fluorescence-transition decay rate in the $\mathrm{MHz}$ range, which leads to  $\xi/\Gamma \sim 10^{-2}$. In microwave optomechanics \cite{Liu2025} report a silicon-carbide mechanical mode with an energy lifetime of $19.9\mathrm{s}$ and a pure-dephasing contribution to the linewidth of approximately $0.28\mathrm{mHz}$, corresponding to an equivalent $\xi/\Gamma$ of order $10^{-2}$ when comparing adjacent phonon-number coherence with mechanical energy damping. Note, that these comparisons illustrate only experimental rate scales by identifying the respective dissipative channels. Properly comparing them with the clockwork jump rates as above requires rigorous and implementation-specific modeling.\\

\section{Enlarging the memory state space $\memoryspace$}
\label{app:memory_size_dependence}
In \cref{app:go_example_qubit_clockwork}, we gave the example of an energy-based feedback policy for two qubit clockworks that outperforms any constant feedback policy.
In the case of energy-based feedback, the parameter-update function chooses between different qubit energies according to the memory state $m$, i.e., $H_{C_{\sidx}}(\memorytorates{\sidx}(m)) = -\frac{\memorytorates{\sidx}(m) }{2}\sigma_{C_{\sidx}}^{z}$. 
In this section, we examine whether the observed advantage increases for higher dimensional memory spaces.
In particular, we study two different generalizations of the feedback policy from \cref{app:go_example_qubit_clockwork}
 to $d$-dimensional memory spaces.

In the first feedback policy $\Pi_{\mathrm{td}}=\{\memoryspace_{\mathrm{td}}, \memoryupdate{}{}{\mathrm{td}},\memorytorates{}_{\mathrm{td}}\}$, the memory state $m\in\memoryspace_{\mathrm{td}}= \{0,\ldots,d-1\}$ records how many out of the last $d-1$ jumps originated from clockwork $C_2$.
This is achieved by a memory-update function 
\begin{equation}
    \memoryupdate{m}{\sidx}{\mathrm{td}}=\begin{cases}
        \max(m-1,0)&\textrm{, if } a=1\\
        \min(m+1,d-1)&\textrm{, if } a=2.
    \end{cases}
\end{equation}

In the second feedback policy $\Pi_{\mathrm{fs}}=\{\memoryspace_{\mathrm{fs}}, \memoryupdate{}{}{\mathrm{fs}},\memorytorates{}_{\mathrm{fs}}\}$, we require $d$ to be a power of two.
The memory state $m\in\memoryspace_{\mathrm{fs}}= \{1\cdots11,1\cdots12,\ldots,2\cdots22\}$ records the entire sequence of the last $\log_2(d)$ jumps.
This is achieved by a memory-update function $\memoryupdate{m}{\sidx}{\mathrm{fs}}$ that shifts all numbers in the string of numbers $m$ by one position to the right, dropping the last number and appending $a$ as the first number, to create a new string of length $\log_2(d)$.

As we expect an optimal feedback policy to be symmetric with respect to both clockworks, we numerically optimize over all possible symmetric parameter-update functions, i.e., all parameter-update functions of the form 
$\memorytorates{2}_{\mathrm{td}}(m)=\memorytorates{1}_{\mathrm{td}}(d-1-m)$ and $\memorytorates{2}_{\mathrm{fs}}(m)=\memorytorates{1}_{\mathrm{fs}}(\bar{m})$, where $\bar{m}$ denotes the string of numbers that is obtained from $m$ by reversing the positions of all ones and twos.
During this optimization, we additionally optimize over all possible choices of the angle $\phi$ in the jump operator $J_{C_\sidx,\sidx} = \sqrt{\Gamma} \ketbra{\phi}{+}_{C_\sidx}$. However, in contrast to the qubit energies, the angle $\phi$ is kept fixed for all memory states $m$.
The resulting signal-to-noise ratio for an integrated current that counts all jumps with equal weights, regardless of which clockwork the jumps occur in and which memory state $m$ the memory is in, are plotted in \cref{fig:robustness_and_mem_dependence} for $2\leq d\leq8$.
Additionally, the right-hand scale shows the advantage this signal-to-noise ratio constitutes when compared to the best constant feedback policy $\Pi_\mathrm{{const}}$ for two qubit clockworks, which is discussed in the main text. 
The corresponding parameters chosen by these feedback policies are listed in \cref{tab:full_sequence_pars,tab:tick_difference_pars}.
In the case of $d=2$, one may easily check that both $\Pi_{\mathrm{fs}}$ and $\Pi_{\mathrm{td}}$ are equivalent to the feedback policy $\Pi_{\mathrm{switch}}$ presented in \cref{app:go_example_qubit_clockwork} and yield an SNR of approximately $2.59\Gamma$, corresponding to an advantage of roughly $9\%$.
We observe that with increasing $d$ the advantage of policy $\Pi_\mathrm{fs}$ increases to roughly $12.5$\% at $d=8$, whereas the advantage of policy $\Pi_{\mathrm{td}}$ seems to vanish.

\begin{table}[]
    \centering
    \begin{tabular}{c||c|c|c|c|c|c|c|c||c}
     $d$ & $\memorytorates{1}_{\mathrm{fs}}(111)/E^*$ & $\memorytorates{1}_{\mathrm{fs}}(112)/E^*$ & $\memorytorates{1}_{\mathrm{fs}}(121)/E^*$ & $\memorytorates{1}_{\mathrm{fs}}(122)/E^*$ & $\memorytorates{1}_{\mathrm{fs}}(211)/E^*$ & $\memorytorates{1}_{\mathrm{fs}}(212)/E^*$ & $\memorytorates{1}_{\mathrm{fs}}(221)/E^*$ & $\memorytorates{1}_{\mathrm{fs}}(222)/E^*$ & $\phi/\phi^*$ \\
     \hline
     2 & 1.200 & 0.883 & - & - & - & - & - & - & 1.008 \\
     4 & 1.317 & 0.933 & 1.152 & 0.782 & - & - & - & - & 1.009 \\
     8 & 1.514 & 0.992 & 1.168 & 0.765 & 1.239 & 0.944 & 1.150 & 0.908 & 1.007 \\
    \end{tabular}
    \caption{Optimized parameters of policy $\Pi_{\mathrm{fs}}$ for different $d$. Note, that we have specified the memory states $m$ assuming $d=8$. For $d=2,4$, the correct memory state $m$ is obtained by cutting off the first (two) numbers. 
    For example, for $m=121$, we have $21$ $(d=4)$ and $1$ $(d=2)$.}
    \label{tab:full_sequence_pars}
\end{table}

\begin{table}[]
    \centering
    \begin{tabular}{c||c|c|c|c|c|c|c|c||c}
     $d$ & $\memorytorates{1}_{\mathrm{td}}(0)/E^*$ & $\memorytorates{1}_{\mathrm{td}}(1)/E^*$ & $\memorytorates{1}_{\mathrm{td}}(2)/E^*$ & $\memorytorates{1}_{\mathrm{td}}(3)/E^*$ & $\memorytorates{1}_{\mathrm{td}}(4)/E^*$ & $\memorytorates{1}_{\mathrm{td}}(5)/E^*$ & $\memorytorates{1}_{\mathrm{td}}(6)/E^*$ & $\memorytorates{1}_{\mathrm{td}}(7)/E^*$ & $\phi/\phi^*$ \\
     \hline
     2 & 1.200 & 0.883 & - & - & - & - & - & - & 1.008 \\
     3 & 1.213 & 1.019 & 0.861 & - & - & - & - & - & 1.005 \\
     4 & 1.205 & 1.023 & 1.003 & 0.857 & - & - & - & - & 1.003 \\
     5 & 1.194 & 1.026 & 1.013 & 0.997 & 0.854 & - & - & - & 1.003 \\
     6 & 1.189 & 1.024 & 1.021 & 1.008 & 0.991 & 0.850 & - & - & 1.002 \\ 
     7 & 1.185 & 1.025 & 1.021 & 1.018 & 1.001 & 0.987 & 0.848 &  - & 1.002 \\
     8 & 1.183 & 1.025 & 1.024 & 1.020 & 1.012 & 0.995 & 0.984 & 0.846 & 1.002 \\
    \end{tabular}
    \caption{{Optimized parameters of policy $\Pi_{\mathrm{td}}$ for different $d$.}}
    \label{tab:tick_difference_pars}
\end{table}

For policy $\Pi_{\mathrm{fs}}$, an increase of the signal-to-noise-ratio with increasing $d$ is expected:
For $d=4$, for example, the feedback policy knows in which clockworks and in which order the last two jumps occurred.
Since the feedback policy can simply ignore all information about the second to last jump, any feedback policy for $d=2$ can also be implemented in the case of $d=4$.
Therefore, the best feedback policy for $d=4$ performs at least equally well as the best feedback policy for $d=2$.
Trivially, the same argument holds for larger $d$.

From \cref{tab:full_sequence_pars}, we observe that the feedback policy $\Pi_{\mathrm{fs}}$ locally behaves similar to $\Pi_\mathrm{switch}$ described in  \cref{app:go_example_qubit_clockwork}.
For most memory states $m$: Given that the next jump occurs in clockwork $C_1$ the energy difference between ground and excited state in clockwork $C_1$ will be increased after the jump.
If on the other hand the next jump occurs in clockwork $C_2$, the energy difference in clockwork $C_1$ will be decreased.

From \cref{tab:tick_difference_pars}, we observe that for policy $\Pi_{\mathrm{td}}$ feedback only takes place at the extreme points $m=0,d-1$. 
For all intermediate memory states $1\leq m\leq d-2$, both clockworks are essentially operated at energy $E^*$, which are the parameters that the best constant feedback policy $\Pi_{\mathrm{const}}$ would choose.
Therefore, as $d$ increases the policy $\Pi_{\mathrm{td}}$ spends on average less time performing any feedback and the signal-to-noise ratio asymptotically approaches that of the best constant feedback policy, as can be seen in \cref{fig:robustness_and_mem_dependence}. 

These findings suggest the feedback policy's performance depends crucially on the degree of connectivity between memory states, which is determined by the feedback policy's memory-update function $\memoryupdate{}{}{}$. Compared to policy $\Pi_{\mathrm{fs}}$, the memory states of policy $\Pi_{\mathrm{td}}$ are less connected in the following sense:
Treat the memory states of policy $\Pi_{\mathrm{td}}$ as the vertices of a directed graph $G_{\mathrm{td}}$, that has an edge (potentially a self-loop) from memory state $m$ to memory state $m'$ iff $\memoryupdate{m}{\sidx}{\mathrm{td}}=m'$ for some $\sidx$. We construct $G_\mathrm{fs}$ analogously.
We observe that both in $G_\mathrm{td}$ and $G_{\mathrm{fs}}$ each vertex, i.e., each memory state, has exactly two outgoing edges.
However, for any two vertices in $G_{\mathrm{fs}}$ there exists a path of length at most $\log_2(d)$ connecting them, whereas the shortest path in $G_{\mathrm{td}}$ connecting vertex $0$ with vertex $d-1$ is of length $d-1$. 
Therefore, we expect that starting in a memory state $m$, the feedback policy $\Pi_\mathrm{td}$ is less likely to leave the neighborhood around $m$ within a fixed time interval, when compared to policy $\Pi_\mathrm{fs}$, potentially hindering its ability to apply the correct feedback.

Before we conclude, let us remark that in this section we have considered only one integrated current, instead of optimizing over all possible ones. 
Which integrated current performs optimally in terms of the signal-to-noise ratio for a given feedback policy, remains an interesting topic for future research.
Furthermore, we expect that the maximal signal-to-noise ratio, achievable by any feedback policy and integrated current, does not grow arbitrarily with $d$, but is determined by the clockworks and allowed feedback mechanism. 
However, studying this in detail requires new, more advanced analytical and numerical tools.

\newpage

}

\end{appendices}
\end{document}